\documentclass{article}
\usepackage{graphicx,graphics,bm,epstopdf}
\usepackage{amsmath,amsfonts,amssymb,hyperref,verbatim}
\usepackage{natbib}
\usepackage[para,online,flushleft]{threeparttable}

\newcommand\dsum{\displaystyle\sum}

\newtheorem{corollary}{Corollary}
\newtheorem{lemma}{Lemma}

\newtheorem{proof}{Proof}

\renewcommand{\thesection}{\arabic{section}}
\renewcommand{\theequation}{\arabic{section}.\arabic{equation}}

\begin{document}

\title{A computationally efficient correlated mixed Probit for credit risk modelling}

\author{Elisa Tosetti\\[4pt]
	\textit{Business School, Brunel University London, Uxbridge UB8 3PH, United Kingdom}\\[4pt]	
	Veronica Vinciotti\\[4pt]
	\textit{Department of Mathematics, Brunel University London, Uxbridge UB8 3PH, United Kingdom}\\[2pt]
     {veronica.vinciotti@brunel.ac.uk}	}

\maketitle

  \begin{abstract}
Mixed Probit models are widely applied in many fields where prediction of a binary
response is of interest. Typically, the random effects are assumed to be independent but this is seldom the case for many real applications. In the credit
risk application considered in this paper, random effects are present at the level of industrial sectors and they are expected to be correlated due to inter-firm credit links inducing dependencies in the
firms' risk to default. Unfortunately, existing inferential procedures for correlated mixed Probit models are computationally very intensive already for a moderate number of effects. Borrowing from the literature on
large network inference, we propose an efficient Expectation-Maximization algorithm for
unconstrained and penalised likelihood estimation and derive the asymptotic standard errors
of the estimates. An extensive simulation study shows that the proposed approach enjoys substantial computational gains
relative to standard Monte Carlo approaches, while still providing accurate parameter estimates. Using data on nearly 64,000 accounts for small and
medium-sized enterprises in the United Kingdom in 2013 across 14 industrial sectors, we find that accounting
for network effects via a correlated mixed Probit model increases significantly the default
prediction power of the model compared to conventional default prediction models, making efficient inferential procedures for these models particularly useful in this field.

\end{abstract}

\section{Introduction}

Discrete choice models with correlated group-specific random effects have a
wide applicability and practical importance in economics and the social
sciences, as they allow to accommodate for unobserved heterogeneity,
over-dispersion, intra- as well as inter-cluster correlation across binary
outcomes. In this paper, we consider the prediction of a firm's risk to default, whereby group random effects  at the level of industrial sectors are to be expected, and, at the same time,  dependencies between and within the industrial sectors are also to be expected due to inter-firm credit links.

Unfortunately, the presence of correlated random effects in mixed models poses
substantial computational challenges, with maximum
likelihood estimation typically requiring the evaluation of a
high-dimensional integral. To overcome these numerical difficulties, various
methods have been proposed in the literature that approximate the likelihood
by Gauss-Hermite quadrature or Monte Carlo integration and then maximize it
by either Newton-Raphson or Expectation-Maximisation algorithms (\citealp{Breslow1993,Schilling2005}). Despite achieving a
computational gain, these methods can still be applied only in the presence
of a limited number of groups because the number of evaluation points in the
Gauss-Hermite quadrature increases exponentially with the number of random
effects. In addition, these approximate Maximum Likelihood (ML) estimates
have been proved to be inconsistent under various conditions, with an
asymptotic bias that can be severe if the variance components are not small (\citealp{Breslow1995}).

An alternative, widely used, approach for estimating mixed models for binary
variables combines Monte Carlo integration with various
Expectation-Maximisation (EM) algorithms, leading to the so-called Monte
Carlo EM algorithm (see, among others, \citealp{Ashford1970,Chib1998,McCulloch1997,Gueorguieva2001}). For the case of a
mixed Probit model with independent random effects, \citet{McCulloch1994}
proposes Monte Carlo versions of the EM algorithm for ML\ estimation based
on the Gibbs sampling. This approach has been extended by
\citet{McCulloch1997} to the more general case of generalised linear mixed
models, by considering a Metropolis-Hastings algorithm at each E-step of the
ML estimation. Similarly, for the case of a mixed Probit model with
correlated random effects, \citet{Chan1997} propose an EM algorithm where
the E-step is made feasible by Gibbs sampling. 
The proposed approach however is
computationally very intensive, as it requires sampling from a multivariate
truncated Normal distribution. In order to deal with this problem,
\citet{Tan2007} propose a non-iterative importance sampling approach to
evaluate the first and the second order moments of a truncated multivariate
normal distribution associated with the Monte Carlo EM algorithm. An
alternative, direct sampling-based, EM algorithm is advanced by
\citet{An2012}, who propose to draw random samples from the prior
Gaussian distribution of random effects. This is computationally easier
than from the posterior distribution, but at the
expense of a higher Monte Carlo error. One limitation of the above Monte
Carlo EM algorithms is that, by combining Monte Carlo simulation with
iterative procedures, they are still computationally very expensive. The
estimation involved in the E-step of the Monte Carlo EM algorithm can
require a prohibitively large amount of time for a large number of
statistical units and already a moderate number of random effects.

In this paper, motivated by a large credit risk application, we propose an EM algorithm for estimation of a mixed Probit
model with correlated random effects that can be adopted for estimation and
prediction from very large data sets and a large number of random effects.
The proposed algorithm relies on efficient approximations of conditional
expectations that simplify the calculation of the moments of a truncated
normal distribution and avoid computationally demanding sampling methods.
Similar approximations have been adopted in the context of graphical models for ordinal (\citealp{Guo2015,Behrouzi2018}) and censored (\citealp{Augugliaro2018}) data but they have not been used in a regression context before. Similar to those approaches, we also propose a penalised version of the likelihood estimator, by applying the graphical lasso approach
(\citealp{Friedman2008}) within the proposed EM algorithm.
Beyond point-wise estimation, standard errors of maximum likelihood
estimates in the context of mixed-effects models are also typically obtained
by time consuming re-sampling approaches. In this paper we exploit the work
by \citet{Louis1982} to derive the observed Fisher information matrix of our
proposed mixed Probit model and thus to obtain the asymptotic standard
errors of the estimates. In doing this, we adopt results by %
\citet{Horrace2015} to calculate the third and fourth moments of univariate
truncated normals which appear in the observed Fished information matrix.  This paper provides a number of contributions to the existing literature.
First, we propose an extension of the literature on non-linear mixed models
to the case of correlated random effects, offering an inferential procedure
that allows to estimate unknown parameters and associated standard errors
also in the presence of very large samples. In doing this, we investigate
ways of overcoming serious computational difficulties that often arise in regression models with correlated binary responses. We also
show how penalised inferential procedures can be applied under this
framework, allowing to cover the case where the number of random effects
exceeds the number of observations.

An extensive simulation study assessing the properties and computational efficiency of our inferential procedure
 shows a good performance of
the proposed approach compared with existing ones. Using data on around 64,000 accounts of unlisted Small and Medium-sized
Enterprises (SMEs) based in the United Kingdom and observed in the year
2013, we find that incorporating inter-firm network dependencies in the form of correlated random effects increases the default prediction power of the credit risk
model compared to conventional ones. 
The remainder of the paper is structured
as follows. Section \ref{data} describes the empirical application on credit risk which motivates this study. Section \ref{sme} introduces our mixed graphical Probit model
and describes the EM algorithm for parameter estimation, with the proposed efficient approximations of the conditional expectations, the inference under penalised likelihood and the derivation of asymptotic standard errors,.
 Section \ref{montecarlo} carries out an extensive simulation study on the proposed method, and
Section \ref{results} describes the results of the proposed approach on real data. Finally, Section \ref{concl} provides some concluding remarks.

\section{Motivating example: credit risk modelling of SMEs } \label{data}
There is nowadays interest in creating default prediction models for Small and Medium-sized Enterprises (SMEs). Academic research into failure prediction has focused almost
exclusively on large companies, i.e. those which are listed on, and priced
by, the market, proposing a wide range of models and methods to assess and
quantify their risk of default.\ On the contrary, there has been a
relatively small number of prior academic studies examining default
prediction and credit scoring models with reference to small, private,
businesses, mostly due to the difficulty in obtaining sufficient and good
quality data in these contexts. These models are likely to be different to
those used for large corporates. For this reason, the recent Basel accords
are now directing the international credit system to pay closer attention to
measuring and managing credit risk of SMEs (\citealp{Sabato2010}).

When modelling credit risk for SMEs, one important feature to be considered
is the fact that companies are not simply independent agents competing for
customers on markets. They are linked by supply-costumer relationships. Some
firms may offer trade credit to other firms, thus establishing inter-firms
credit links (\citealp{Battiston2007}). Clearly, firms interact with each
others because they exchange items of value, such as information, goods,
services, and money. For example, the output of some firms (sub-contractors)
are input for some other firms. In addition, some firms may extend trade
credit to other firms, thus creating some sort of inter-firms credit links (%
\citealp{Battiston2007}). Interdependence amongst firms' default can also
arise because they share part of the management team and hence are subject
to similar investment decisions, or because firms react similarly to
external shocks such as a rise in the interest rate (\citealp{Andrews2003}).
Under this framework, the failure of a firm is likely to increase the
probability of failure of connected firms, giving raise to clustered
fluctuations in the number of failed firms.

Despite the importance of inter-firm links in determining firms'
performance, only few studies have looked at the role of interaction in
determining firms' default and clusters of default, with the majority of
these studies focusing on identifying the conditions under which local
failures can result in bankruptcies across the network (\citealp{Gatti2006}), or exploring whether firms that issue more trade
credit are more likely to experience a debtor failure (\citealp{Jacobson2013}). Yet fewer studies have considered incorporating
information on firms' interdependence into a default prediction model. Among
these, \citet{Barro2010} have proposed a model of contagion that associates
the economic relationship of sectors of the economy and the geographical
proximity of each pair of firms in a network of firms, whereas %
\citet{Barreto2013} have developed a measure of local risk of default using
ordinary kriging from data on 9 million Brazilian SMEs observed in 2010.
After including this measure as an additional explanatory variable in a
logistic credit scoring model, the authors showed that the performance of
the model improved considerably.

It is well known that the financial performance of companies are in part
driven by sector- and area-specific attributes, linked for example to
heterogeneity across industries in accounting policies or local trends in
demand (see, for example, \citealp{Kukuk2013}). For this reason, mixed
discrete choice models have been widely adopted to predict firm financial
distress for large corporations (see, among others, \citealp{Jones2004,Kukuk2013}), with few studies also specific to SMEs (see, for example %
\citealp{Alfo2005}). Differently from this literature and considering the
importance of inter-firm dependencies discussed above, in this paper, we
allow group effects to be correlated by assigning them a non-diagonal
covariance matrix. Under this framework, the dependence relationship of the
binary outcomes (default) is induced by the underlying Gaussian graphical
model on the random effects. In particular, we assume that the risk of
default for one company follows a Probit regression specification with
correlated group random effects, where groups are given by all companies
operating in the same sector of economic activity and located in the same
region.

We exploit a rich data set from a large financial institution covering around
64,000 accounts of unlisted firms based in the United Kingdom and observed
in the year 2013.\ These are companies that have no more than 250 employees,
a turnover smaller than \pounds 25.9 million, and a balance sheet total of
no more than \pounds 12.9 million. In line with other studies, we define
failure as entry into liquidation, administration or receivership. The
accounts analysed for failed companies are the last set of accounts filed in
the year prior to insolvency. The companies are spread over a total of 59 geographical areas, defined
using the NUTS3 classification, and across 14 broad sectors (divisions)\ of economic activity. In our model, the sectors will appear as random effects, whereas the geographical areas as the sampling units.

The data set contains a set of financial variables extracted from the
accounts of firms, as well as non-financial information, that are often
included in conventional default prediction models (see, among others, %
\citealp{Altman2007,Altman2010,Carling2007,Campbell2008,Jacobson2013}). In terms of firm-specific
financial variable, we include a set of financial ratios that cover the
areas of profitability, liquidity, leverage, coverage and activity (%
\citealp{Altman2007}). Profitability is the ability of the firm to generate
sufficient profits or returns, liquidity measures the ability of the firm to
meet its short-term obligations, leverage refers to the relative amount of
debt and other obligations of the firm, coverage is the risk inherent in
lending to the business in long-term, while activity is the level of
efficiency of a business. As for the non-financial indicators, we consider
variables linked to the age and size of the companies. We expect a higher
risk of default for newly formed companies that decreases with the age of
the company, and that is particularly high in the years immediately after an
initial \textquotedblleft honeymoon period\textquotedblright\ of around two
years. Finally, we have matched information on the postal district of the
trading address with data on latitude and longitude and other geographical
information extracted from the UK Office of National Statistics, to
calculate covariates at the aggregated level and account for systematic
risk. In particular, we include the NUTS3-level Gross Domestic Product, as a
proxy for the economic conditions of the area where the company operates.
Table \ref{tab1} lists the financial ratios included in our analysis grouped
according to the financial indicators and the non-financial ones, including company
characteristics and aggregate variables.

Table \ref{tab3} provides a set of descriptive statistics for the variables
included in our model, for failed and non-failed companies. As expected,
companies that failed have on average worse leverage and liquidity
indicators than firms that did not fail; they are smaller in
size and younger and more frequently fall in the age risk group. It is
interesting to observe that both trade debt and trade credit ratios have
higher values for defaulted companies. This result is supported by the
literature on trade credit which shows evidence that financially distressed
small companies not only have higher levels of trade debt supplied to
customers but also of trade credit obtained from suppliers (%
\citealp{CARBO-VALVERDE2016}).

\begin{table}
\caption{Credit risk data: definition of financial ratios, non-financial indicators and
aggregate variables. \label{tab1}}
\centering%
\begin{tabular}{|c|c|}
\hline
{\footnotesize Variable} & {\footnotesize Accounting ratio
category} \\ \hline
\multicolumn{2}{|c|}{\footnotesize Financial indicators} \\ \hline
\multicolumn{1}{|l|}{\footnotesize Total liabilities/total assets} &
{\footnotesize Leverage} \\
\multicolumn{1}{|l|}{\footnotesize Networth/total liabilities} &
{\footnotesize Leverage} \\
\multicolumn{1}{|l|}{\footnotesize Cash/total assets} & {\footnotesize %
Liquidity} \\
\multicolumn{1}{|l|}{\footnotesize Current liabilities/current assets} &
{\footnotesize Liquidity} \\
\multicolumn{1}{|l|}{\footnotesize Trade credit/total liabilities} &
{\footnotesize Liquidity} \\
\multicolumn{1}{|l|}{\footnotesize Trade debt/total assets} & {\footnotesize %
Liquidity} \\
\multicolumn{1}{|l|}{\footnotesize Retained profits/total assets} &
{\footnotesize Profitability} \\
\multicolumn{1}{|l|}{\footnotesize Account receivable/total liabilities} &
{\footnotesize Activity} \\ \hline
\multicolumn{2}{|c|}{\footnotesize Non-financial characteristics} \\ \hline
\multicolumn{1}{|l|}{\footnotesize Size} & \multicolumn{1}{|l|}%
{\footnotesize Total assets (logs)} \\
\multicolumn{1}{|l|}{\footnotesize Age} & \multicolumn{1}{|l|}{\footnotesize %
Age from the date of incorporation (logs)} \\
\multicolumn{1}{|l|}{\footnotesize Age risk} & \multicolumn{1}{|l|}{%
{\footnotesize 1 if }{\footnotesize 3 $\leq$ age $\leq$ 9}{\footnotesize \ years}} \\
\multicolumn{1}{|l|}{\footnotesize Local GDP} & \multicolumn{1}{|l|}%
{\footnotesize Gross Domestic Product in the NUTS3} \\ \hline
\end{tabular}
\end{table}%

\begin{table}
\caption{Credit risk data: descriptive statistics for non-failed and failed companies on
training sample. \label{tab3}}
\centering
\begin{tabular}{|c|c|c|cc|}
\hline
& \multicolumn{2}{|c|}{\footnotesize Non failed} & \multicolumn{2}{|c|}%
{\footnotesize Failed} \\ \hline
{\footnotesize Variable} & {\footnotesize Mean} & {\footnotesize Standard Error}
& {\footnotesize Mean} & \multicolumn{1}{|c|}{\footnotesize Standard Error} \\
\hline
\multicolumn{1}{|l|}{\footnotesize Total liabilities/total assets} &
\multicolumn{1}{|r|}{\footnotesize 0.817} & \multicolumn{1}{|r|}%
{\footnotesize 1.243} & \multicolumn{1}{|r}{\footnotesize 1.278} &
\multicolumn{1}{r|}{\footnotesize 1.851} \\
\multicolumn{1}{|l|}{\footnotesize Networth/total liabilities} &
\multicolumn{1}{|r|}{\footnotesize 6.315} & \multicolumn{1}{|r|}%
{\footnotesize 22.461} & \multicolumn{1}{|r}{\footnotesize 3.155} &
\multicolumn{1}{r|}{\footnotesize 15.173} \\
\multicolumn{1}{|l|}{\footnotesize Cash/total assets} & \multicolumn{1}{|r|}%
{\footnotesize 0.333} & \multicolumn{1}{|r|}{\footnotesize 0.348} &
\multicolumn{1}{|r}{\footnotesize 0.377} & \multicolumn{1}{r|}{\footnotesize %
0.380} \\
\multicolumn{1}{|l|}{\footnotesize Current liabilities/current assets} &
\multicolumn{1}{|r|}{\footnotesize 1.826} & \multicolumn{1}{|r|}%
{\footnotesize 5.283} & \multicolumn{1}{|r}{\footnotesize 2.386} &
\multicolumn{1}{r|}{\footnotesize 5.806} \\
\multicolumn{1}{|l|}{\footnotesize Trade credit/total liabilities} &
\multicolumn{1}{|r|}{\footnotesize 0.197} & \multicolumn{1}{|r|}%
{\footnotesize 0.302} & \multicolumn{1}{|r}{\footnotesize 0.225} &
\multicolumn{1}{r|}{\footnotesize 0.350} \\
\multicolumn{1}{|l|}{\footnotesize Trade debt/total assets} &
\multicolumn{1}{|r|}{\footnotesize 0.155} & \multicolumn{1}{|r|}%
{\footnotesize 0.231} & \multicolumn{1}{|r}{\footnotesize 0.162} &
\multicolumn{1}{r|}{\footnotesize 0.263} \\
\multicolumn{1}{|l|}{\footnotesize Retained profits/total assets} &
\multicolumn{1}{|r|}{\footnotesize -0.030} & \multicolumn{1}{|r|}%
{\footnotesize 0.594} & \multicolumn{1}{|r}{\footnotesize -0.216} &
\multicolumn{1}{r|}{\footnotesize 1.039} \\
\multicolumn{1}{|l|}{\footnotesize Account receivable/total liabilities} &
\multicolumn{1}{|r|}{\footnotesize 0.006} & \multicolumn{1}{|r|}%
{\footnotesize 0.029} & \multicolumn{1}{|r}{\footnotesize 0.004} &
\multicolumn{1}{r|}{\footnotesize 0.025} \\ \hline
\multicolumn{1}{|l|}{\footnotesize Size} & \multicolumn{1}{|r|}%
{\footnotesize 12.311} & \multicolumn{1}{|r|}{\footnotesize 2.899} &
\multicolumn{1}{|r}{\footnotesize 10.489} & \multicolumn{1}{r|}%
{\footnotesize 2.484} \\
\multicolumn{1}{|l|}{\footnotesize Age} & \multicolumn{1}{|r|}{\footnotesize %
2.382} & \multicolumn{1}{|r|}{\footnotesize 0.927} & \multicolumn{1}{|r}%
{\footnotesize 1.757} & \multicolumn{1}{r|}{\footnotesize 0.873} \\
\multicolumn{1}{|l|}{\footnotesize Age risk} & \multicolumn{1}{|r|}%
{\footnotesize 0.346} & \multicolumn{1}{|r|}{\footnotesize 0.476} &
\multicolumn{1}{|r}{\footnotesize 0.445} & \multicolumn{1}{r|}{\footnotesize %
0.497} \\
\multicolumn{1}{|l|}{\footnotesize Local GDP} & \multicolumn{1}{|r|}%
{\footnotesize 10.229} & \multicolumn{1}{|r|}{\footnotesize 0.447} &
\multicolumn{1}{|r}{\footnotesize 10.213} & \multicolumn{1}{r|}%
{\footnotesize 0.433} \\ \hline
\end{tabular}
\end{table}

In the next section, we formalise the proposed mixed Probit model with correlated random effects and describe an inferential procedure that is computationally efficient for data such as that described in this section, for which existing mixed probit models are prohibitively slow.

\section{Efficient mixed Probit model with correlated random effects} \label{sme}
\subsection{The model}
Consider a sample of $N_{r}$ companies located in region $r$, with$\ r=1,2,...,R$. Let $y_{ir}$
be the dichotomous variable equal to 1 when company
$i$ located in region $r$  defaults. Let $G$ be the number of industrial sectors. Using the latent response
model, we assume that $y_{ir}$ is generated by thresholding the latent
variable $y_{ir}^{\ast }$ that follows the Gaussian mixed model:%
\begin{eqnarray}
y_{ir}^{\ast } &=&\bm{\beta }^{\prime }\mathbf{x}_{ir}+\mathbf{z}%
_{ir}^{\prime }\mathbf{u}_{r}+\varepsilon _{ir},  \label{1} \\
y_{ir} &=&1\text{ if }y_{ir}^{\ast }\geq 0\text{, 0 otherwise,}  \notag
\end{eqnarray}%
where $\mathbf{x}_{ir}$ is a $K$-dimensional vector of explanatory
variables, $\bm{\beta }$ is a $K$-dimensional vector of unknown
parameters, $\mathbf{u}_{r}=\left( u_{1r},u_{2r},...,u_{Gr}\right) ^{\prime} $ is a
$G$-dimensional vector of Gaussian random errors with $\mathbf{z}_{ir} $ being a $G$-dimensional vector of (known) loadings,  and
$\varepsilon _{ir} $ are Gaussian random errors. We assume that\ $\mathbf{u}
_{r}$ and $\varepsilon _{ir}$ satisfy the following conditions for all $r$:
\begin{eqnarray*}
E\left( \varepsilon _{ir}\right) &=&0,E\left( \varepsilon _{ir}^{2}\right)
=1,\text{ for }i=1,2,...,N_{r}, \\
E\left( \varepsilon _{ir}\varepsilon _{js}\right) &=&0,\text{ for }i\neq
j=1,2,...,N_{r};r,s=1,2,...,R, \\
E\left( \mathbf{u}_{r}\mathbf{u}_{r}^{\prime }\right) &=&\mathbf{\Sigma }%
_{G}, \\
E\left( \mathbf{u}_{r}\mathbf{u}_{s}^{\prime }\right) &=&\mathbf{0},\text{
for }r\neq s, \\
E(\mathbf{u}_{r}\varepsilon _{is}) &=&\mathbf{0}\text{ for }r,s=1,2,...,R,
\end{eqnarray*}%
where $\mathbf{\Sigma }_{G}$ is a positive definite matrix with $\sigma
_{gh} $ the ($g,h$) off-diagonal element and $\sigma _{g}^{2}$ the $g$th
diagonal element. In stacked form model (\ref{1}) can be written as:%
\begin{equation*}
\mathbf{y}_{r}^{\ast }=\mathbf{X}_{r}\bm{\beta }+\mathbf{Z}_{r}\mathbf{u}%
_{r}+\bm{\varepsilon }_{r},
\end{equation*}%
where $\mathbf{y}_{r}^{\ast }=\left( y_{1r}^{\ast },y_{2r}^{\ast
},...,y_{N_{r},r}^{\ast }\right) ^{\prime }$, $\mathbf{X}_{r}=\left( \mathbf{%
x}_{1r},\mathbf{x}_{2r},...,\mathbf{x}_{Nr}\right) ^{\prime }$ $\bm{%
\varepsilon }_{r}=\left( \varepsilon _{1r},\varepsilon _{2r},...,\varepsilon
_{Nr}\right) ^{\prime }$ and $\mathbf{Z}_{r}$ is an $N_{r}\times G$ matrix.
In addition, $\mathbf{y}_{r}^{\ast }$ has covariance:%
\begin{equation}
\mathbf{\Sigma }_{r}=\mathbf{Z}_{r}\mathbf{\Sigma }_{G}\mathbf{Z}_{r}^{\prime }+%
\mathbf{I}_{N_{r}}.  \label{9}
\end{equation}%
The model above allows for group effects that vary across $R$ and $G$,
although the dependencies are only allowed across the $G$ dimension.
\subsection{Inference} \label{inference}
The interest is in estimating the regression parameters, $\bm{\beta }$,
as well as the dependence structure among the $G$ groups, given by the
elements of the precision matrix, $\mathbf{\Phi }_{G}=\mathbf{\Sigma }%
_{G}^{-1}$. As also remarked by the graphical modelling literature,
estimating the elements of the precision matrix allows to assess whether any
two units are conditionally independent given all other units (%
\citealp{Lauritzen1996}), thus providing a network of dependencies at the
level of random effects. Accordingly, let $\bm{\vartheta }=(\bm{%
\beta },vech(\mathbf{\Phi }_{G}))$ be the vector of unknown parameters in
the above model, and note that the observed data, $\mathbf{y}=\left( \mathbf{%
y}_{1}^{\prime },\mathbf{y}_{2}^{\prime },...,\mathbf{y}_{R}^{\prime
}\right) ^{\prime }$, is a function of the unobserved variables $\mathbf{y}%
^{\ast }=\left( \mathbf{y}_{1}^{\ast \prime },\mathbf{y}_{2}^{\ast \prime
},...,\mathbf{y}_{R}^{\ast \prime }\right) ^{\prime }$ and $\mathbf{u}%
=\left( \mathbf{u}_{1}^{\prime },\mathbf{u}_{2}^{\prime },...,\mathbf{u}%
_{R}^{\prime }\right) ^{\prime }$. The log-likelihood of the observed data
is given by:%
\begin{equation}
l(\bm{\vartheta })=\log \int f_{\mathbf{y,y}^{\ast },\mathbf{u}}\left(
\mathbf{y},\mathbf{y}^{\ast },\mathbf{u}|\bm{\vartheta }\right) d\mathbf{%
y}^{\ast }d\mathbf{u}.  \label{l}
\end{equation}%
The integral in (\ref{l}) makes it difficult to maximize $l(\bm{%
\vartheta })$ directly, so an EM algorithm for computing ML estimates can be
adopted, by maximizing the conditional expectation of the log-likelihood
function for the complete data given the observed data $\mathbf{y}$.
Treating $\mathbf{y}$, $\mathbf{y}^{\ast }$ and $\mathbf{u}$ as the complete
data, and $\mathbf{y}$ as the incomplete data, we have%
\begin{equation}
l(\bm{\vartheta })=\log f_{\mathbf{y},\mathbf{y}^{\ast },\mathbf{u}%
}\left( \mathbf{y},\mathbf{y}^{\ast },\mathbf{u}|\bm{\vartheta }\right)
-\log f_{\mathbf{y}^{\ast },\mathbf{u}|\mathbf{y}}\left( \mathbf{y}^{\ast },%
\mathbf{u}|\bm{y,\vartheta }\right) ,  \label{lfeta}
\end{equation}%
where $\log f_{\mathbf{y},\mathbf{y}^{\ast },\mathbf{u}}\left( \mathbf{y},%
\mathbf{y}^{\ast },\mathbf{u}|\bm{\vartheta }\right) $ is the
log-likelihood function for the complete data, namely%
\begin{gather}
\log f_{\mathbf{y},\mathbf{y}^{\ast },\mathbf{u}}\left( \mathbf{y},\mathbf{y}%
^{\ast },\mathbf{u}|\bm{\vartheta }\right) =\log \left[ f\left( \mathbf{u%
}\right) f\left( \mathbf{y}^{\ast },\mathbf{y}|\mathbf{u}\right) \right]
\notag \\
\approx \frac{R}{2}\ln \left\vert \mathbf{\Phi }_{G}\right\vert -\frac{1}{2}%
\sum_{r=1}^{R}\mathbf{u}_{r}^{\prime }\mathbf{\Phi }_{G}\mathbf{u}_{r}-\frac{%
1}{2}\sum_{r=1}^{R}\left( \mathbf{y}_{r}^{\ast }-\mathbf{X}_{r}\bm{\beta
}-\mathbf{Z}_{r}\mathbf{u}_{r}\right) ^{\prime }\left( \mathbf{y}_{r}^{\ast
}-\mathbf{X}_{r}\bm{\beta }-\mathbf{Z}_{r}\mathbf{u}_{r}\right).  \notag
\end{gather}%
Taking conditional expectations given $\mathbf{y}$ on both sides of (\ref%
{lfeta}) yields:%
\begin{eqnarray}
l(\bm{\vartheta }) &=&E\left[ \log f_{\mathbf{y},\mathbf{y}^{\ast },%
\mathbf{u}}\left( \mathbf{y},\mathbf{y}^{\ast },\mathbf{u}|\bm{\vartheta
}\right) |\mathbf{y}\right] -E\left[ \log f_{\mathbf{y}^{\ast },\mathbf{u}|%
\mathbf{y}}\left(\mathbf{y}^{\ast },\mathbf{u}|\bm{y,\vartheta }\right) |%
\mathbf{y}\right]  \label{QH} \\
&=&Q\left( \bm{\vartheta }\right) -H\left( \bm{\vartheta }\right) ,
\notag
\end{eqnarray}%
where%
\begin{eqnarray}
Q\left( \bm{\vartheta }\right) &\approx &\frac{R}{2}\ln \left\vert
\mathbf{\Phi }_{G}\right\vert -\frac{1}{2}{\rm Tr}\left\{ \mathbf{\Phi }_{G}\frac{1%
}{R}\sum_{r=1}^{R}E\left( \mathbf{u}_{r}\mathbf{u}_{r}^{\prime }\mathbf{|y}%
_{r}\right) \right\}  \label{Qfunction} \\
&-&\frac{1}{2}\sum_{r=1}^{R}E\left[ \left( \mathbf{y}_{r}^{\ast }-\mathbf{X}%
_{r}\bm{\beta }-\mathbf{Z}_{r}\mathbf{u}_{r}\right) ^{\prime }\left(
\mathbf{y}_{r}^{\ast }-\mathbf{X}_{r}\bm{\beta }-\mathbf{Z}_{r}\mathbf{u}%
_{r}\right) \mathbf{|y}_{r}\right] .  \notag
\end{eqnarray}

The $Q$ function is the main ingredient of the EM algorithm. Let $\bm{%
\hat{\vartheta}}^{(m)}$ denote the estimate of $\mathbf{\Theta } $ after the
$m$th iteration. Then the E and M steps of the $(m+1)$th iteration are given
by:
\begin{description}
\item[E-Step] (Expectation step) Compute $Q\left( \bm{\vartheta |\hat{%
\vartheta}}^{(m)}\right) =E\left[ \log f_{\mathbf{y},\mathbf{y}^{\ast },%
\mathbf{u}}\left( \mathbf{y},\mathbf{y}^{\ast },\mathbf{u}|\bm{\hat{%
\vartheta}}^{(m)}\right) |\mathbf{y}\right] $

\item[M-Step] (Maximisation step): Compute $\bm{\hat{\vartheta}}%
^{(m+1)}=\arg \max Q\left( \bm{\vartheta |\hat{\vartheta}}^{(m)}\right) $%
.
\end{description}
These steps are iterated until convergence is achieved. For $R>>G$, the
first-order conditions for $\bm{\beta }$ and $\mathbf{\Phi }_{G}$ in the
$M$-step are:%
\begin{eqnarray}
\bm{\hat{\beta}}^{(m+1)} &=&\left( \sum_{r=1}^{R}\mathbf{X}_{r}^{\prime }%
\mathbf{X}_{r}\right) ^{-1}\sum_{r=1}^{R}\mathbf{X}_{r}^{\prime }\left[
E\left( \mathbf{y}_{r}^{\ast }|\mathbf{y}_{r}\right) -\mathbf{Z}_{r}E\left(
\mathbf{u}_{r}\mathbf{|y}_{r}\right) \right] ,  \label{bhat1} \\
\mathbf{\Phi }_{G}^{(m+1)} &=&\left[ \frac{1}{R}\sum_{r=1}^{R}E\left(
\mathbf{u}_{r}\mathbf{u}_{r}^{\prime }\mathbf{|y}_{r}\right) \right] ^{-1}.
\label{teta1}
\end{eqnarray}%
Hence, the $M$-step alternates between estimation of $\bm{\beta }$
using (\ref{bhat1}) and estimation of $\mathbf{\Phi }_{G}$ using (\ref{teta1}%
). At each step, the new estimate of $\mathbf{\Phi }_{G}$ uses the previous
value of $\bm{\hat{\beta}}$\ and the new value of $\mathbf{\hat{\Phi}}%
_{G}$ is used to update $\bm{\hat{\beta}}$. \citet{Meng1993} showed that
iterating between these two equations in the EM algorithm provides
convergence to the true ML\ estimates. However, the above expressions depend
on the unknown quantities $E\left( \mathbf{u}_{r}\mathbf{|y}_{r}\right) $
and $E\left( \mathbf{u}_{r}\mathbf{u}_{r}^{\prime }\mathbf{|y}_{r}\right) $.
In the following, we propose an approximation of conditional expectations $%
E\left( \mathbf{u}_{r}\mathbf{|y}_{r}\right) $ and $E\left( \mathbf{u}_{r}%
\mathbf{u}_{r}^{\prime }\mathbf{|y}_{r}\right) $ and show how this can be
adopted to simplify the EM\ algorithm.

\subsection{Approximating conditional expectations} \label{EM}
Using the law of iterated expectations and the theorem on conditional
normals, $E\left( \mathbf{u}_{r}\mathbf{|y}_{r}\right) $ and $E\left(
\mathbf{u}_{r}\mathbf{u}_{r}^{\prime }\mathbf{|y}_{r}\right) $ are typically
calculated by
\begin{equation}
E\left( \mathbf{u}_{r}\mathbf{|y}_{r}\right) =\mathbf{\Sigma }_{G}\mathbf{Z}%
_{r}^{\prime }\mathbf{\Sigma }_{r}^{-1}\left[ E\left( \mathbf{y}_{r}^{\ast }|%
\mathbf{y}_{r}\right) -\mathbf{X}_{r}\bm{\beta }\right] ,  \label{Eu1}
\end{equation}%
\begin{eqnarray}
E\left( \mathbf{u}_{r}\mathbf{u}_{r}^{\prime }\mathbf{|y}_{r}\right) &=&%
\mathbf{\Sigma }_{G}\mathbf{Z}_{r}^{\prime }\mathbf{\Sigma }_{r}^{-1}E\left[
\left( \mathbf{y}_{r}^{\ast }-\mathbf{X}_{r}\bm{\beta }\right) \left(
\mathbf{y}_{r}^{\ast }-\mathbf{X}_{r}\bm{\beta }\right) ^{\prime }|%
\mathbf{y}_{r}\right] \mathbf{\Sigma }_{r}^{-1}\mathbf{Z}_{r}\mathbf{\Sigma }%
_{G}  \notag \\
&&+\mathbf{\Sigma }_{G}-\mathbf{\Sigma }_{G}\mathbf{Z}_{r}^{\prime }\mathbf{%
\Sigma }_{r}^{-1}\mathbf{Z}_{r}\mathbf{\Sigma }_{G},  \label{euugiveny}
\end{eqnarray}
following Appendix B and \citet{Chan1997}.

From the above expressions it is clear that $E\left( \mathbf{u}_{r}%
\mathbf{|y}_{r}\right) $ and $E\left( \mathbf{u}_{r}\mathbf{u}_{r}^{\prime }%
\mathbf{|y}_{r}\right) $ depend on the first two moments of a multivariate
truncated normal distribution, namely, $E\left( \mathbf{y}_{r}^{\ast }|%
\mathbf{y}_{r}\right) $ and $E\left[ \left( \mathbf{y}_{r}^{\ast }-\mathbf{X}%
_{r}\bm{\beta }\right) \left( \mathbf{y}_{r}^{\ast }-\mathbf{X}_{r}%
\bm{\beta }\right) ^{\prime }|\mathbf{y}_{r}\right] $.\ A number of
authors have proposed algorithms for direct estimation or approximation of
moments of multivariate truncated normal distributions (see, among others, %
\citealp{Tallis1961,Lee1979,Leppard1989}). Other authors
have proposed a Markov Chain Monte Carlo approach that consists of randomly
generating a sequence of samples from the multivariate truncated normal
distribution and then approximating the first two moments by the empirical
conditional moments from these samples (\citealp{Kotecha1999,Chan1997,Chib1998,Abegaz2015}). Although this method is faster than
direct estimation of the moments, it is still computationally very demanding
for large scale problems. A recent strand of literature has proposed to
approximate the first and second moments of a multivariate truncated normal
distribution through an iterative procedure within the M-step (\citealp{Guo2015,Behrouzi2018,Augugliaro2018}), leading to a computationally much
faster approach than any previous methods. Exploiting this literature, we
consider a mean field approximation of the second moments, namely for $i\neq
j$ and for all $r=1,2,...,R$:%
\begin{equation}
E\left[ \left( y_{ir}^{\ast }-\bm{\beta }^{\prime }\mathbf{x}%
_{ir}\right) \left( y_{jr}^{\ast }-\bm{\beta }^{\prime }\mathbf{x}%
_{jr}\right) |\mathbf{y}_{r}\right] \approx E\left[ \left( y_{ir}^{\ast }-%
\bm{\beta }^{\prime }\mathbf{x}_{ir}\right) |\mathbf{y}_{r}\right] E%
\left[ \left( y_{jr}^{\ast }-\mathbf{\beta }^{\prime }\mathbf{x}_{jr}\right)
|\mathbf{y}_{r}\right] .  \label{ass1}
\end{equation}%
Hence, 
once controlled for the observed values in $\mathbf{y}_{r}$ and the regressors $\mathbf{X}_{r}$, $%
y_{ir}^{\ast }$ and $y_{jr}^{\ast }$ become decoupled. In the simulation
section we will show good properties of our proposed estimator with that
based on the slower Monte Carlo EM procedures, that do not make the above
approximation (see Section \ref{montecarlo}). Under (\ref{ass1}), in order
to compute (\ref{Eu1})-(\ref{euugiveny}),
we only need to find $E\left( y_{ir}^{\ast }|\mathbf{y}_{r}\right) $\ and $%
E\left( y_{ir}^{\ast 2}|\mathbf{y}_{r}\right) $. To this end, first write
the first and second conditional moments as follows:%
\begin{eqnarray}
E\left( y_{ir}^{\ast }|\mathbf{y}_{r}\right) &=&E\left[ E\left( y_{ir}^{\ast
}|\mathbf{y}_{-i,r}^{\ast },y_{ir}\right) |\mathbf{y}_{r}\right] ,
\label{c1} \\
E\left( y_{ir}^{\ast 2}|\mathbf{y}_{r}\right) &=&E\left[ E\left(
y_{ir}^{\ast 2}|\mathbf{y}_{-i,r}^{\ast },y_{ir}\right) |\mathbf{y}_{r}%
\right] ,  \label{c2}
\end{eqnarray}%
where $\mathbf{y}_{-i,r}^{\ast }=(y_{1r}^{\ast },y_{2r}^{\ast
},...,y_{i-1,r}^{\ast },y_{i+1,r}^{\ast },...,y_{N_{r}r}^{\ast })^{\prime }$. Noting that $\mathbf{y}_{r}^{\ast }$ is a vector of jointly normal
variables with mean zero and covariance $\mathbf{\Sigma }_{r}$, and
exploiting the theorem on conditional normals, we obtain that the
conditional distribution of $y_{ir}^{\ast }$ given $\mathbf{y}_{-i,r}^{\ast
} $ has mean and variance respectively given by%
\begin{eqnarray*}
\tilde{\mu}_{ir} &=&\bm{\beta }^{\prime }\mathbf{x}_{ir}+\mathbf{\Sigma }%
_{r,i,-i}\mathbf{\Sigma }_{r,-i,-i}^{-1}\left( \mathbf{y}_{-i,r}^{\ast }-%
\mathbf{X}_{-i,r}\bm{\beta }\right) , \\
\tilde{\sigma}_{ir}^{2} &=&\sigma _{ir}^{2}-\mathbf{\Sigma }_{r,i,-i}\mathbf{%
\Sigma }_{r,-i,-i}^{-1}\mathbf{\Sigma }_{r,-i,i},
\end{eqnarray*}%
where $\sigma _{ir}^{2}$ is the ($i,i$)th element of $\mathbf{\Sigma }_{r}$.
Replacing the above expressions in the equation for the mean and second
moment of truncated normals (see Appendix A) 
we obtain the following expressions for the first conditional moment (\ref{c1}) and the second conditional moment (\ref{c2}):%
\begin{eqnarray}
E\left( y_{ir}^{\ast }-\bm{\beta }^{\prime }\mathbf{x}_{ir}|\mathbf{y}%
_{r}\right) &\!=\!&\mathbf{\Sigma }_{r,i,-i}\mathbf{\Sigma }_{r,-i,-i}^{-1}%
E\left( \mathbf{y}_{-i,r}^{\ast }-\mathbf{X}_{-i,r}\bm{\beta} |\mathbf{y}%
_{r}\right)
+\rho _{1,ir}\tilde{\sigma}_{ir},  \label{conde} \\
E\left[ \left( y_{ir}^{\ast }-\bm{\beta }^{\prime }\mathbf{x}%
_{ir}\right) ^{2}|\mathbf{y}_{r}\right] &\!=\!&\mathbf{\Sigma }_{r,i,-i}\mathbf{%
\Sigma }_{r,-i,-i}^{-1}E\left[ \left( \mathbf{y}_{-i,r}^{\ast }-\mathbf{X}%
_{-i,r}\bm{\beta }\right) \left( \mathbf{y}_{-i,r}^{\ast }-\mathbf{X}%
_{-i,r}\bm{\beta }\right) ^{\prime }|\mathbf{y}_{r}\right] \mathbf{%
\Sigma }_{r,-i,-i}^{-1}\mathbf{\Sigma }_{r,i,-i}  \notag \\
&+&\tilde{\sigma}_{ir}^{2}+2\rho _{1,ir}\tilde{\sigma}_{ir}\mathbf{\Sigma }%
_{r,i,-i}\mathbf{\Sigma }_{r,-i,-i}^{-1}E\left( \mathbf{y}_{-i,r}^{\ast }-%
\mathbf{X}_{-i,r}\mathbf{\beta} |\mathbf{y}_{r}\right) +\rho _{2,ir}\tilde{\sigma}_{ir}^{2}
\notag \\
&+&\left( \bm{\beta }^{\prime }\mathbf{x}_{ir}\right) ^{2}-2\bm{%
\beta }^{\prime }\mathbf{x}_{ir}E\left( y_{ir}^{\ast }|\mathbf{y}_{r}\right)
,  \label{conde2a}
\end{eqnarray}%
where $\rho _{1,ir}$ and $\rho _{2,ir}$ are defined in Appendix A. 
The above equations show that there exists a recursive
relationship between the elements in $E\left( y_{ir}^{\ast }-\bm{\beta }%
^{\prime }\mathbf{x}_{ir}|\mathbf{y}_{r}\right) $ and \\ $E\left[ \left(
\mathbf{y}_{r}^{\ast }-\mathbf{X}_{r}\bm{\beta }\right) \left( \mathbf{y}%
_{r}^{\ast }-\mathbf{X}_{r}\bm{\beta }\right) ^{\prime }|\mathbf{y}_{r}%
\right] $ and offer an iterative procedure for estimating these quantities.
More specifically: Let $E\left( y_{jr}^{\ast }-\bm{\beta }^{\prime }%
\mathbf{x}_{jr}|\mathbf{y}_{r}\right) ^{(h)}$ and $E\left[ \left(
y_{jr}^{\ast }-\bm{\beta }^{\prime }\mathbf{x}_{jr}\right) ^{2}|\mathbf{y%
}_{r}\right] ^{(h)}$, for all $j$, be the estimates of $E\left( y_{jr}^{\ast }-\bm{%
\beta }^{\prime }\mathbf{x}_{jr}|\mathbf{y}_{r}\right) $ and $E\left[ \left(
y_{jr}^{\ast }-\bm{\beta }^{\prime }\mathbf{x}_{jr}\right) ^{2}|\mathbf{y%
}_{r}\right] $, respectively, at the $h$th stage in the $M$-step. We plug these into
the right hand side of (\ref{conde})-(\ref{conde2a}) to compute new values of
$E\left( y_{ir}^{\ast }-\bm{\beta }^{\prime }\mathbf{x}_{ir}|\mathbf{y}%
_{r}\right) $ and $E\left[ \left( y_{ir}^{\ast }-\bm{\beta }^{\prime }%
\mathbf{x}_{ir}\right) ^{2}|\mathbf{y}_{r}\right] $ (inner iterations).
After convergence is reached, let $E\left( y_{ir}^{\ast }-\mathbf{%
\beta }^{\prime }\mathbf{x}_{ir}|\mathbf{y}_{r}\right) ^{(h)\ast }$ and $E%
\left[ \left( y_{ir}^{\ast }-\bm{\beta }^{\prime }\mathbf{x}_{ir}\right)
^{2}|\mathbf{y}_{r}\right] ^{(h)\ast }$ be the final estimates. We plug these into (\ref{bhat1}) to obtain a new estimate of $\bm{\beta }$ and to compute (\ref{euugiveny}) that enters in (\ref%
{teta1}) for estimation of $\mathbf{\Phi }_{G}$ (outer iterations).
With the
new $\bm{\beta }$ and $\mathbf{\Phi }_{G}$, we recompute $E\left(
y_{ir}^{\ast }-\bm{\beta }^{\prime }\mathbf{x}_{ir}|\mathbf{y}%
_{r}\right) $ and $E\left[ \left( y_{ir}^{\ast }-\bm{\beta }^{\prime }%
\mathbf{x}_{ir}\right) ^{2}|\mathbf{y}_{r}\right] $ ready for another round
of inner iterations.  Note however that convergence for the inner iterations
is not necessary; in fact, inner iterations can be reduced to a single round
of computation.

According to the iterative procedure just described, the matrix inverse, $\mathbf{%
\Sigma }_{r,-i,-i}^{-1}$, for $i=1,2,...,N_{r}$, needs to be computed at
each iteration of the EM procedure. Although the matrix can be rather large, given that it has size $(N_r-1) \times (N_r-1)$, a simplified expression can be obtained by noting that:%
\begin{equation*}
\mathbf{\Sigma }_{r,-i,-i}=\mathbf{Z}_{r,-i}\mathbf{\Sigma }_{G}\mathbf{Z}%
_{r,-i}^{\prime }+\mathbf{I}_{N_{r}-1},
\end{equation*}%
and, using the matrix inversion lemma:%
\begin{equation*}
\mathbf{\Sigma }_{r,-i,-i}^{-1}=\mathbf{I}_{N_{r}-1}-\mathbf{Z}_{r,-i}\left(
\mathbf{\Sigma }_{G}^{-1}+\mathbf{Z}_{r,-i}^{\prime }\mathbf{Z}%
_{r,-i}\right) ^{-1}\mathbf{Z}_{r,-i}^{\prime }.
\end{equation*}%
Hence, $\mathbf{\Sigma }_{r,-i,-i}^{-1}$ involves computing only the inverse
of $G$-dimensional matrices. This shows the power of using a mixed model approach, whereby dependencies are captured at the lower-dimensional space of the random effects. 

In addition, when $N_{r}$ is particularly large, such as in our real application, we found it computationally beneficial, and not detrimental to the resulting estimators, to replace the expectations (\ref{Eu1})-(\ref{euugiveny}) with the group averages of expectations of the latent variables, that is
\begin{eqnarray}
E\left( u_{gr}\mathbf{|y}_{r}\right) &\approx &\frac{1}{m_{gr}}%
\dsum\limits_{i\in g}\left( E\left( y_{ir}^{\ast }|\mathbf{y}_{r}\right) -%
\bm{\beta }^{\prime }\mathbf{x}_{ir}\right) ,  \label{Eu} \\
E\left( u_{gr}u_{hr}\mathbf{|y}_{r}\right) &\approx &\frac{1}{m_{gr}m_{hr}}%
\dsum\limits_{i\in g;j\in h}E\left[ \left( y_{ir}^{\ast }-\bm{\beta }%
^{\prime }\mathbf{x}_{ir}\right) \left( y_{jr}^{\ast }-\bm{\beta }%
^{\prime }\mathbf{x}_{jr}\right) |\mathbf{y}_{r},\right]  \label{Euu}
\end{eqnarray}%
where $m_{gr}$ is the number of units belonging to group $g$ and located in region $r$ and $\dsum\limits_{i\in g}$ indicates the sum over all units belonging to group $g$ and located in region $r$. The above estimator is widely adopted to proxy random effects (\citealp{Hsiao2003}), also in the context of
cross sectionally dependent panels (\citealp{Moscone2016}). 

Finally, further computational efficiency can be achieved by applying penalised maximum
likelihood, as described in the next subsection.

\subsection{Penalised maximum likelihood estimation} \label{penlik}
When the condition $R>>G$ does not hold, unconstrained maximum likelihood
estimation of $\mathbf{\Phi }_{G}$ is not feasible. In this case, we add an $%
L_{1}$-norm penalty term to the log-likelihood and optimise the penalised
likelihood:%
\begin{equation*}
l_{1}(\bm{\vartheta })=\log \int f_{\mathbf{y,y}^{\ast },\mathbf{u}%
}\left( \mathbf{y},\mathbf{y}^{\ast },\mathbf{u}|\bm{\vartheta }\right) d%
\mathbf{y}^{\ast }d\mathbf{u}-\rho _{G}\left\Vert \mathbf{\Phi }%
_{G}\right\Vert _{1},
\end{equation*}%
where $\rho_G$ is a tuning parameter controlling the degree of sparsity of
the underlying network and $\left\Vert \mathbf{.}\right\Vert _{1}$ is the $%
L_{1}$ norm on the off-diagonal entries of the precision matrix. When $\rho_G$ is large enough, some coefficients in $\mathbf{%
\Phi }_{G}$ are shrunken to zero, resulting in the removal of the
corresponding links in the underlying network. Noting that\ the part of $%
\log f_{\mathbf{y},\mathbf{y}^{\ast },\mathbf{u}}\left( \mathbf{y},\mathbf{y}%
^{\ast },\mathbf{u}|\bm{\vartheta }\right)$ that depends on $%
\mathbf{\Sigma }_{G}^{-1}$ is the log-likelihood of a multivariate normal,%
\begin{equation*}
Q_{1}\left( \bm{\vartheta |\hat{\vartheta}}^{(m)}\right) =-\frac{R}{2}%
\ln \left\vert \mathbf{\Sigma }_{G}\right\vert -\frac{1}{2}{\rm Tr}\left\{ \mathbf{%
\Sigma }_{G}^{-1}\frac{1}{R}\sum_{r=1}^{R}E\left( \mathbf{u}_{r}\mathbf{u}%
_{r}^{\prime }\mathbf{|y}_{r}\right) \right\},
\end{equation*}%
and following the same line of reasoning as in Section \ref{inference}, we
consider the penalised estimation problem for $\mathbf{\Phi }_{G}$ within
the M-step by optimizing%
\begin{equation}
Q_{1,pen}\left( \bm{\vartheta |\hat{\vartheta}}^{(m)}\right) =\frac{R}{2}%
\ln \left\vert \mathbf{\Phi }_{G}\right\vert -\frac{1}{2}Tr\left\{ \mathbf{%
\Phi }_{G}\frac{1}{R}\sum_{r=1}^{R}E\left( \mathbf{u}_{r}\mathbf{u}%
_{r}^{\prime }\mathbf{|y}_{r}\right) \right\} -\rho _{G}\left\Vert \mathbf{%
\Phi }_{G}\right\Vert _{1}.  \label{q1pen}
\end{equation}%
Hence, we alternate between estimation of $\bm{\beta }$ using (\ref%
{bhat1}) and estimation of $\mathbf{\Phi }_{G}$ using (\ref{q1pen}), for which efficient graphical lasso implementations can be used (\citealp{Friedman2008}).

The regularization parameter $\rho _{G}$ defines the level of sparsity of
the associated network $\mathbf{\hat{\Phi}}_{G}$. A number of criteria are
available in the penalised likelihood literature for the selection of this
parameter, such as the Bayesian Information Criteria (BIC). This and most
other methods are based on the likelihood function of the observed data,
which, for our model, is given by (\ref{QH}). \citet{Ibrahim2008}, however,
suggest to use only the $Q$-function in (\ref{Qfunction}) for calculation of
the likelihood. This is more efficient, as the $Q$-function is a direct
output of the EM algorithm, whereas the $H$-function would need to be
calculated separately.

\subsection{Standard errors approximation} \label{stderr}
Calculating standard errors of estimates requires knowledge of the
information matrix associated to the log-likelihood function of the observed
data, known as the observed information matrix. However, this also involves
computation of the $H$-function in (\ref{QH}), which is not a direct output
of the EM\ iterations. Following \citet{Louis1982}, it is possible to
compute the observed information matrix by exploiting the complete data
gradient and curvature. In particular, let $B\left( \mathbf{y}|\bm{%
\vartheta }\right) =\dfrac{\partial ^{2}l\left( \bm{\vartheta }\right) }{%
\partial \vartheta_i\partial \vartheta_j} $ be the
partial second derivatives of the observed data log-likelihood and $S\left( \mathbf{y,y}^{\ast },\mathbf{u}%
|\bm{\vartheta }\right) =\dfrac{\partial \log f_{\mathbf{y,y}^{\ast },%
\mathbf{u}}\left( \mathbf{y},\mathbf{y}^{\ast },\mathbf{u}|\bm{\vartheta
}\right) }{\partial \vartheta}$ and $B\left( \mathbf{y,y}^{\ast },%
\mathbf{u}|\bm{\vartheta }\right) =\dfrac{\partial ^{2}\log f_{\mathbf{%
y,y}^{\ast },\mathbf{u}}\left( \mathbf{y},\mathbf{y}^{\ast },\mathbf{u}|%
\bm{\vartheta }\right) }{\partial \vartheta_i\partial
\vartheta_j}$ be the gradient and second derivative of the complete data
log-likelihood, respectively. It is possible to show that:%
\begin{align}
B\left( \mathbf{y}|\bm{\vartheta }\right) & =E\left[ B\left( \mathbf{y,y}%
^{\ast },\mathbf{u}|\bm{\vartheta }\right) |\mathbf{y}\right] +E\left[
S\left( \mathbf{y,y}^{\ast },\mathbf{u}|\bm{\vartheta }\right) S\left(
\mathbf{y,y}^{\ast },\mathbf{u}|\bm{\vartheta }\right) ^{\prime }|%
\mathbf{y}\right]  \label{B} \\
& -E\left[ S\left( \mathbf{y,y}^{\ast },\mathbf{u}|\bm{\vartheta }%
\right) |\mathbf{y}\right] E\left[ S\left( \mathbf{y,y}^{\ast },\mathbf{u}|%
\bm{\vartheta }\right) |\mathbf{y}\right] ^{\prime }.  \notag
\end{align}%
Hence, by exploiting the law of iterated expectations as well as the approximation (%
\ref{ass1}), it is also possible to compute efficiently all terms appearing on the right
hand side of (\ref{B}). In Appendix  D 
we provide finite expressions
for the elements of $B\left( \mathbf{y}|\bm{\vartheta }\right) $.

\section{Simulation study} \label{montecarlo}
In order to assess the performance of our proposed approach, we consider a simulation study using the following data generating process:%
\begin{eqnarray*}
y_{ir}^{\ast } &=&\beta x_{ir}+\mathbf{z}_{ir}^{\prime }\mathbf{u}%
_{r}+\varepsilon _{ir},i=1,2,...,N_{r};r=1,2,...R,  \label{mc1} \\
y_{ir} &=&1\text{ if }y_{ir}^{\ast }\geq 0\text{, 0 otherwise,}  \label{mc2}
\end{eqnarray*}%
where we set $\beta =1$, $\mathbf{x}_{r}=\left(
x_{1r},x_{1r},...,x_{N_{r}r}\right) \sim N(\mathbf{0},\mathbf{\Sigma }_{X})$
and $\mathbf{u}_{r}\sim N(\mathbf{0},\mathbf{\Sigma }_{G})$. To generate $%
\mathbf{\Sigma }_{G}$, we start from $\mathbf{\Theta }_{G}=\mathbf{\Sigma }%
_{G}^{-1}$ and assume that $\theta _{gh,G}\sim Bin\left( 1,\frac{3}{G}%
\right) $ for $g=1,...,G,h=g,...,G$. We then let $\mathbf{D}$ be the
Choleski decomposition of $\mathbf{\Sigma }_{G}$, namely $\mathbf{\Sigma }%
_{G}=\mathbf{DD}^{\prime }$, and we generate $\mathbf{u}_{r}=\mathbf{%
D\epsilon }_{r}$, where $\bm{\epsilon }_{r}=\left( \epsilon
_{1r},\epsilon _{2r},...,\epsilon _{Gr}\right) ^{\prime }$, with $\epsilon
_{ir}\sim IDN(0,1)$. We finally obtain $\mathbf{\Sigma }_{r}$ by applying
formula (\ref{9}). We generate $\mathbf{\Sigma }_{X}$ following the same
procedure.

We carry out two sets of experiments, one with $R=200$ (case of large $R$),
where we compute our proposed estimator, which we call \texttt{mixed
graphical Probit}, and one with $R=50$ where we compute a penalised version
of our estimator. In both experiments we also compute the conventional mixed
Probit with uncorrelated random effects. We take $N_{r}=N=50,100,250$ and
vary $G$ depending on $N_{r}$, from $G=10$ to $G=125$. Each experiment was
replicated $50$ times. In a separate experiment we also carry out a
comparison of our estimator with the Monte Carlo EM estimator by %
\citet{Chan1997}, in terms of performance of estimators and computational
time. Due to the high computational cost of the Monte Carlo EM approach, for
this experiment we have selected smaller values of $G$ ($G\leq 25$). This
comparison is important because the Monte Carlo EM estimator by %
\citet{Chan1997} does not rely on the conditional approximation (\ref{ass1}%
). For the same combinations of $N$ and $R$, we also compare the properties
and computational time of the mixed graphical Probit estimator using (\ref%
{Eu1})-(\ref{euugiveny}) with those of the same estimator based on their
approximations (\ref{Eu})-(\ref{Euu}).

A number of statistics are used to assess the performance of our estimators.
We first report the Receiver Operating Characteristic (ROC) curve for the
predicted outcomes, plotting percentage of non-zero outcomes correctly
predicted as non-zero versus the percentage of zeros incorrectly predicted
as non-zeros, as the classification threshold varies between 0 and 1. To
this end, we generate a testing sample with the same Monte Carlo design as
above, and employ the parameters estimated in the training sample to
calculate predictions. As for the estimation of the slope parameter, $\beta $%
, we report bias and Root Mean Squared Error (RMSE), given by $%
1/50\sum_{s=1}^{50}{\hat{\beta}}_{s}-\beta $, and $\sqrt{1/50\sum_{s=1}^{50}%
\left( \hat{\beta}_{s}-\beta \right) ^{2}}$, respectively. Under penalised
ML\ estimation, we select the regularization parameter $\rho _{G}$ with the
value closest to the true sparsity level. This is only possible in a
simulation setting and allows our results to not depend on the specific
choice of model selection criterion. In addition, we summarise the recovery
of the network structure across the whole path of regularization parameters,
by reporting the corresponding ROC curve. This plots the true positive rate,
i.e. percentage of non-zeros in the estimated precision
matrix $\mathbf{\Phi }_{G}$, that is detected links, correctly estimated as non-zero, versus the
false positive rate, i.e. percentage of zeros incorrectly estimated as
non-zeros, as the tuning parameter, $\rho _{G}$, varies.

The results are reported in Figure \ref{fig1}-\ref{fig4} and Table \ref{tab2}.
Figures \ref{fig1}-\ref{fig2} show the ROC curves for the predicted
outcomes and precision matrix, respectively, estimated by maximum likelihood
and penalised maximum likelihood for varying $N$ and $G$. As expected, the
performance of the mixed graphical Probit estimator improves as $N$
increases for fixed $R$\ and $G$, while it deteriorates as $G$ rises,
holding $N$ and $R$ constant. This result can be explained by looking at the
main features of our model. In fact, as $N$ increases we have more and more
observations to estimate the unknown parameters $\beta $ and $\mathbf{\Phi }%
_{G}$, while when $G$ increases we have more and more parameters to estimate.

Figures \ref{fig3}-\ref{fig4} compare the ROC curve of the predicted
outcomes of our proposed estimator against the conventional mixed Probit
estimator, for the large $R$ and short $R$ scenarios. For all combinations
of $N$, $G$ and $R$ the mixed graphical Probit outperforms the conventional
mixed Probit in predicting correctly the outcome variable. The improvement
in performance seems to be more important when $G$ is small relative to $N$
and when $N$ is large.

Table \ref{tab2} reports the bias, RMSE\ and computational time for the
proposed approach using (\ref{Eu1})-(\ref{euugiveny}) versus the same
estimator based on their approximations (\ref{Eu})-(\ref{Euu}) and finally
versus the full Monte Carlo EM estimator by \citet{Chan1997}, which does not
make the approximation in (\ref{ass1}). It is interesting to observe that
the three estimators have a small bias and RMSE, and that these decrease as $%
N$ rises, while their performance slightly deteriorates as the number of
groups ($G$) increases. Comparing the results in Column (I) and (II), the
computational time of the estimator based on (\ref{Eu})-(\ref{Euu}) is
significantly smaller than that of the estimator based on  (%
\ref{Eu1})-(\ref{euugiveny}), thus supporting the use of group averages of
conditional expectations to proxy random effects. The fact that the bias and
RMSE\ of the estimators in Column (I)-(II) are of comparable size with that
in Column (III) indicates that the approximation in (\ref{ass1}), adopted
both in Column (I) and (II), does not significantly affect the properties of
our estimators. However, the difference in the computational time between
the graphical mixed Probit estimators in Column (I)-(II) and the full Monte
Carlo EM\ estimator in Column (III) is striking, with the mixed graphical
Probit carrying out one estimation in few seconds across all experiments,
against a computational time that can be as long as few minutes in the case
of the Monte Carlo EM\ algorithm.

\begin{figure}
\centering
{\includegraphics[width=0.8\textwidth]{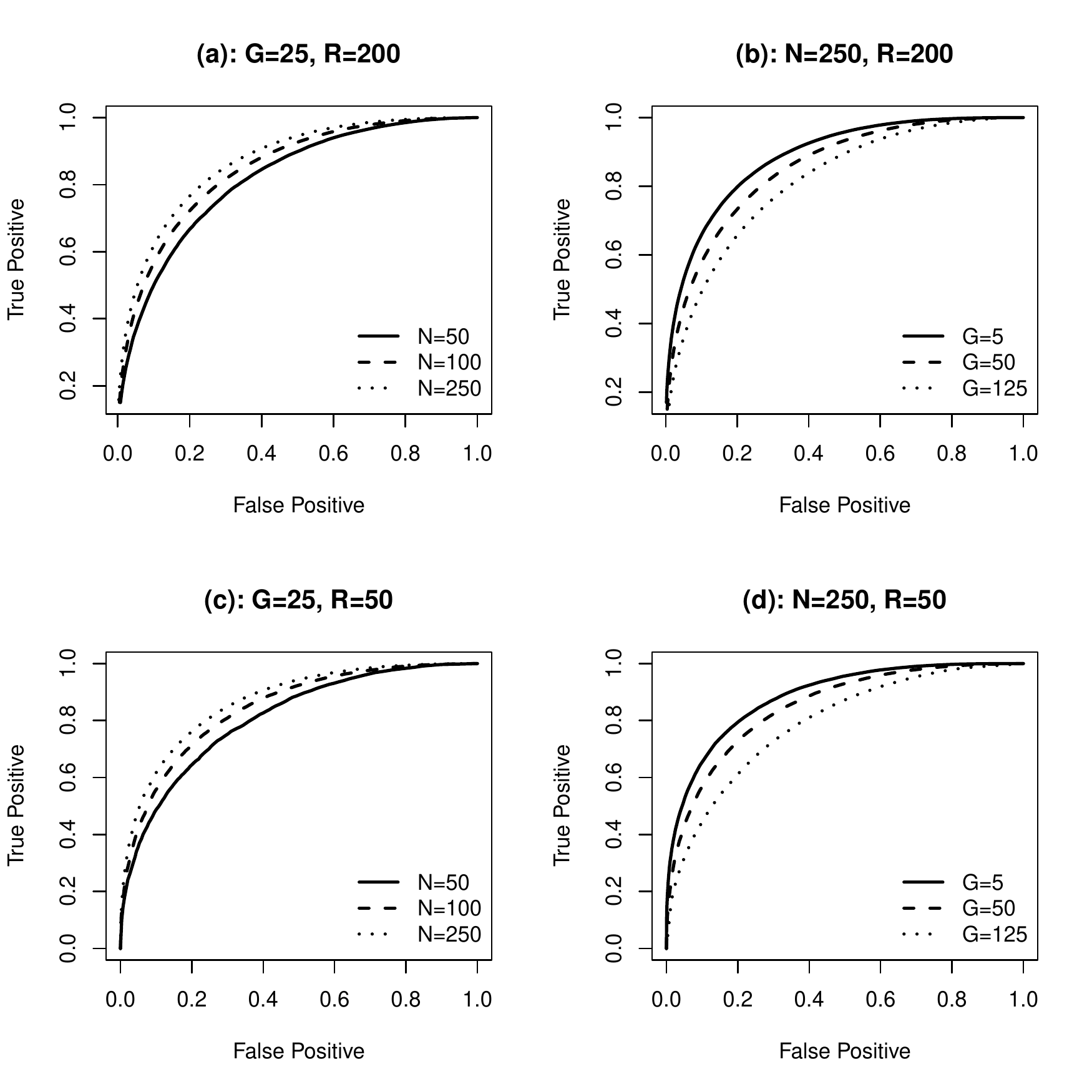} }
\caption{\label{fig1} Simulation study: ROC curves on predicted outcomes
on test set for varying $N$ and $G$, with parameters estimated on the
training set under maximum likelihood (top) and penalised maximum likelihood
(bottom).}
\end{figure}

\begin{figure}
\centering
{\includegraphics[width=0.8\textwidth]{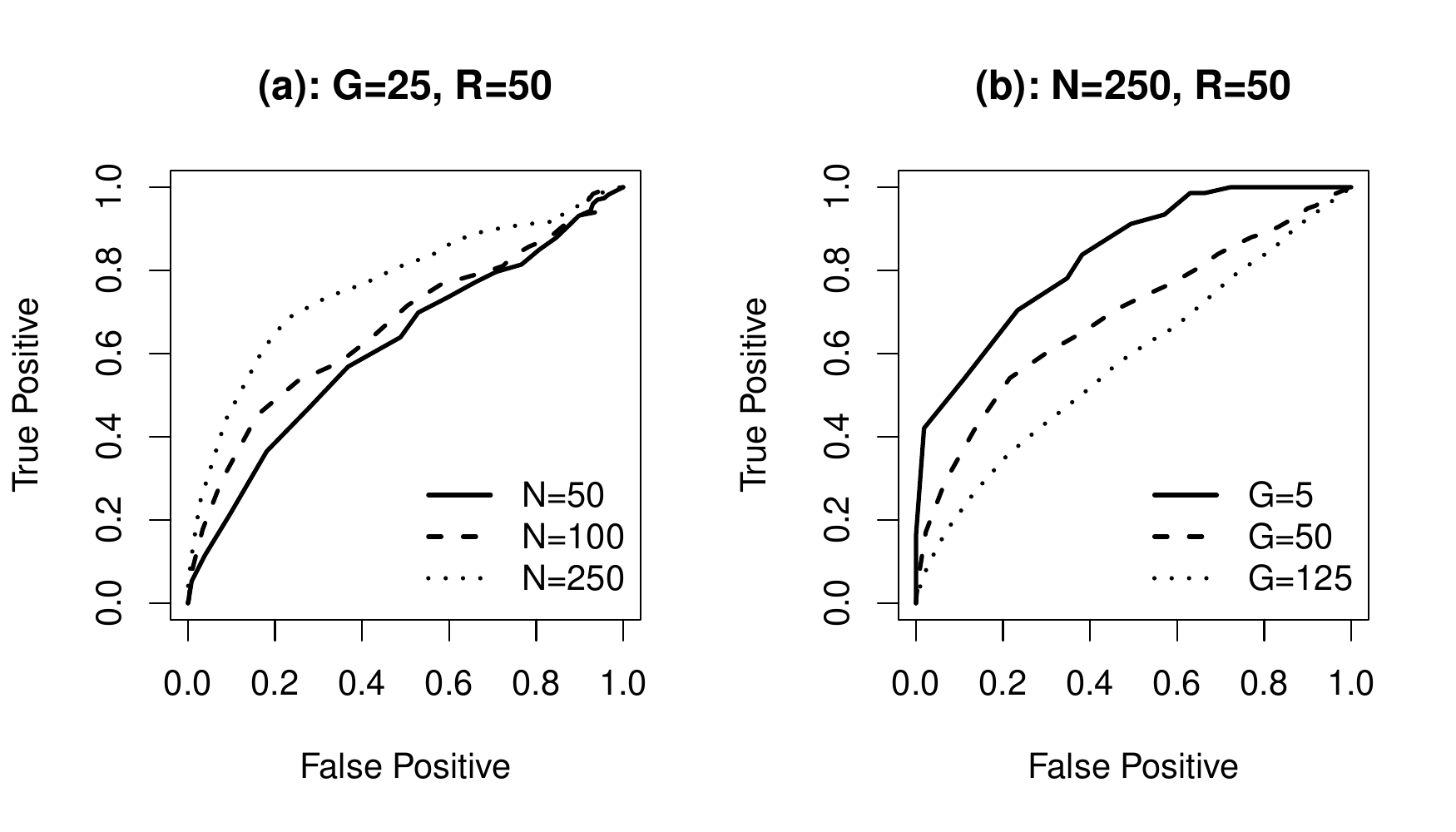} }
\caption{\label{fig2} Simulation study: ROC curves of network discovery for varying N and G under penalised maximum likelihood estimation.}
\end{figure}

\begin{figure}
\centering
{\includegraphics[width=0.8\textwidth]{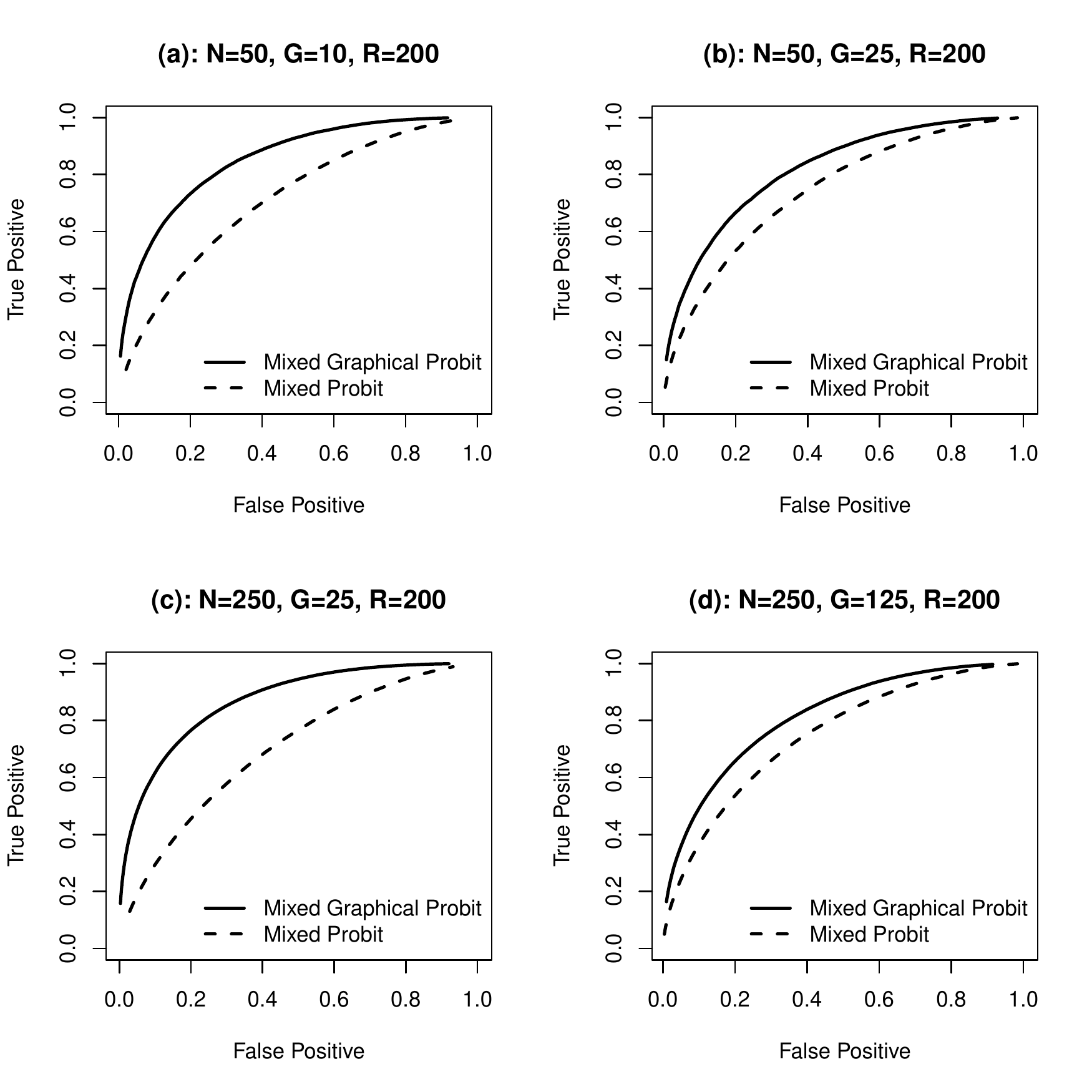} }
\caption{\label{fig3} Simulation study: ROC curves on predicted outcomes
on test set for varying $N$ and $G$ using the mixed graphical Probit
(maximum likelihood) and the mixed Probit with uncorrelated effects.}
\end{figure}

\begin{figure}
\centering
{\includegraphics[width=0.8\textwidth]{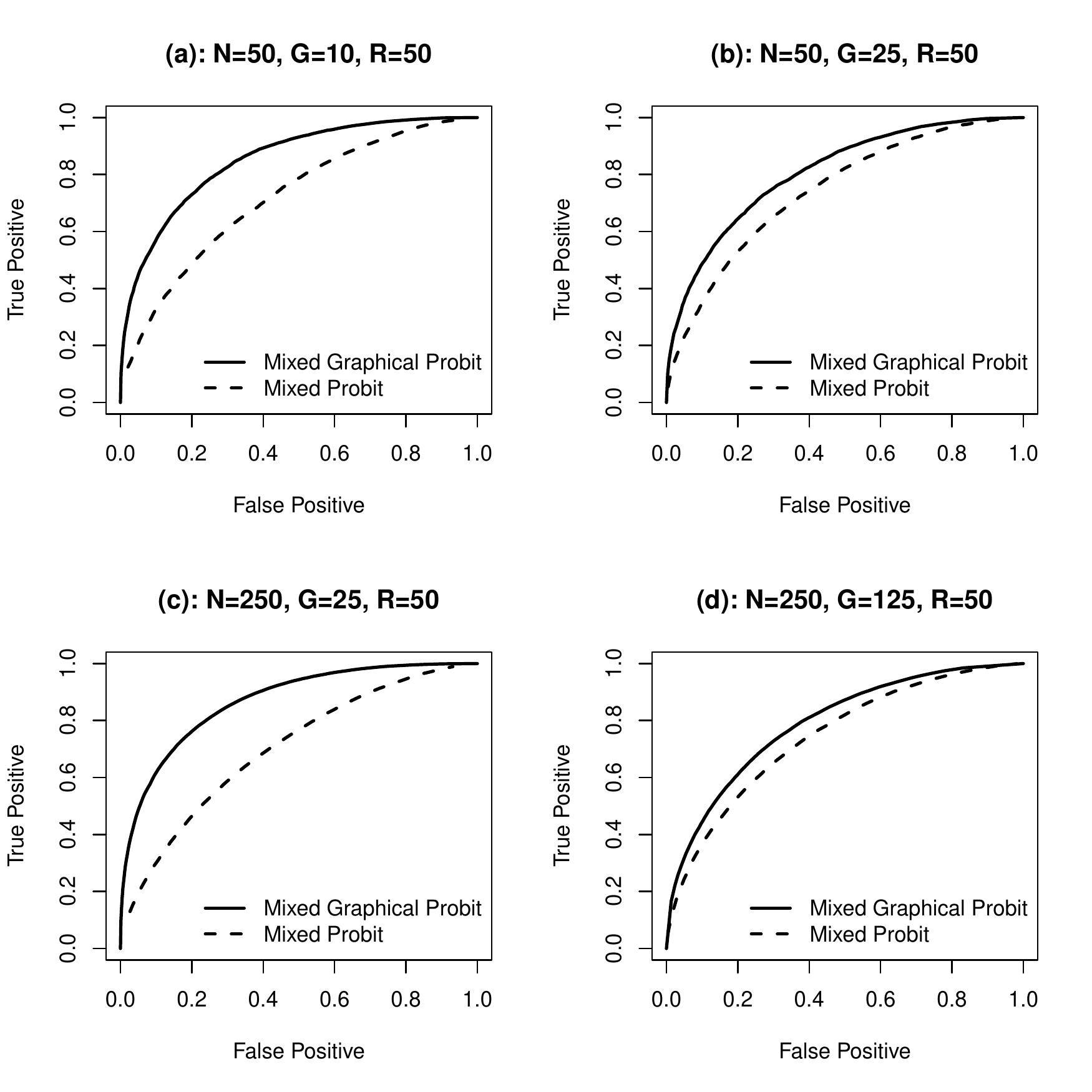} }
\caption{\label{fig4} Simulation study: ROC curves on predicted outcomes
on test set for varying $N$ and $G$ using mixed graphical Probit (penalised
likelihood) and mixed Probit with uncorrelated effects.}
\end{figure}

\begin{table}
\caption{Simulation study: properties of mixed graphical Probit using approximation (\ref{ass1}) and further using (\ref{Eu})-(\ref{Euu}) in place of (\ref{Eu1})-(\ref{euugiveny}), compared with the full Monte Carlo EM estimator. The average computational
time for carrying out one estimation is reported, expressed in seconds. \label{tab2}}
\centering%
\resizebox{\textwidth}{!}{
\begin{tabular}{|ccc|ccc|ccc|c|c|c|}
\hline
&  &  & \multicolumn{3}{|c|}{(I): {\footnotesize Mixed graphical Probit}} &
\multicolumn{3}{|c|}{(II): {\footnotesize Mixed graphical Probit}} &
\multicolumn{3}{|c|}{(III): {\footnotesize Mixed graphical Probit }} \\
&  &  & \multicolumn{3}{|c|}{{\footnotesize using approximations (\ref{ass1})-(\ref{Eu})-(\ref{Euu})}}
& \multicolumn{3}{|c|}{{\footnotesize using approximation (\ref{ass1})}}
& \multicolumn{3}{|c|}{\footnotesize using full Monte Carlo EM} \\ \hline
{\footnotesize N} & {\footnotesize G} & {\footnotesize R} & {\footnotesize %
Bias} & {\footnotesize RMSE} & {\footnotesize Time} &
{\footnotesize Bias} & {\footnotesize RMSE} & {\footnotesize Time}
& {\footnotesize Bias} & {\footnotesize RMSE} & {\footnotesize Time}\\ \hline
{\footnotesize 50} & {\footnotesize 10} & {\footnotesize 200} &
{\footnotesize -0.0021} & {\footnotesize 0.0209} & {\footnotesize 3.0} &
{\footnotesize -0.0080} & {\footnotesize 0.0226} & {\footnotesize 23.4} &
{\footnotesize 0.0266} & {\footnotesize 0.0328} & {\footnotesize 54.0} \\
{\footnotesize 50} & {\footnotesize 25} & {\footnotesize 200} &
{\footnotesize -0.0771} & {\footnotesize 0.0812} & {\footnotesize 4.1} &
{\footnotesize -0.0939} & {\footnotesize 0.0973} & {\footnotesize 305.8} &
{\footnotesize -0.0704} & {\footnotesize 0.0758} & {\footnotesize 46.8} \\
{\footnotesize 100} & {\footnotesize 10} & {\footnotesize 200} &
{\footnotesize -0.0012} & {\footnotesize 0.0149} & {\footnotesize 6.6} &
{\footnotesize -0.0018} & {\footnotesize 0.0149} & {\footnotesize 13.5} &
{\footnotesize 0.0183} & {\footnotesize 0.0247} & {\footnotesize 149.7} \\
{\footnotesize 100} & {\footnotesize 25} & {\footnotesize 200} &
{\footnotesize -0.0056} & {\footnotesize 0.0169} & {\footnotesize 6.2} &
{\footnotesize -0.0087} & {\footnotesize 0.0181} & {\footnotesize 153.8} &
{\footnotesize 0.0199} & {\footnotesize 0.0253} & {\footnotesize 128.3} \\
{\footnotesize 250} & {\footnotesize 10} & {\footnotesize 200} &
{\footnotesize -0.0018} & {\footnotesize 0.0108} & {\footnotesize 21.9} &
{\footnotesize -0.0010} & {\footnotesize 0.0107} & {\footnotesize 25.9} &
{\footnotesize 0.0093} & {\footnotesize 0.0130} & {\footnotesize 1586.5} \\
{\footnotesize 250} & {\footnotesize 25} & {\footnotesize 200} &
{\footnotesize 0.0001} & {\footnotesize 0.0115} & {\footnotesize 28.2} &
{\footnotesize 0.0054} & {\footnotesize 0.0129} & {\footnotesize 97.3} &
{\footnotesize 0.0198} & {\footnotesize 0.0222} & {\footnotesize 1665.2} \\
\hline
\end{tabular}}
\end{table}

\section{Credit risk Probit model with correlated effects} \label{results}

We now employ the proposed approach to estimate a default prediction model for SMEs based on the data described in Section \ref{data}.
To assess the performance of the classifier, we randomly spit the sample into two groups:
40,000 companies are used for estimation (training sample) and the remaining
accounts for testing the prediction accuracy of the model (hold-out sample).
In particular, we compare the prediction performance and estimated
parameters of a conventional credit risk model (mixed and non-mixed) with
that of a credit risk model that incorporates network effects.

Table \ref{tabcr} shows the estimated regression coefficients and standard
errors for the proposed mixed graphical Probit using maximum likelihood
(Column (I)), compared with those of a conventional mixed Probit, with
uncorrelated random effects (Column (II)). Standard errors for the mixed
graphical Probit have been calculated using the observed information matrix
(see Appendix D).

Focusing on Column (I), the coefficient attached to cash over total assets
is statistically significant with a negative sign, indicating that companies
with higher cash reserves relative to current assets are less likely to
default. The results also show a negative and statistically significant
impact for the variable ``retained profits on total assets'': the higher the
net profits with respect to the investments made, the lower the probability
for the firm to go bankrupt. The variable trade debt has a negative and
significant coefficient, meaning that the higher the money a company is
expected to receive from other companies as a result of trade, the less
likely the company is to default. Looking at the non-financial variables,
the coefficients attached to ``size'' and ``age'' indicate that, as expected,
larger and older companies have lower probabilities of default. However,
companies aged between 3 and 9 years have a relatively higher likelihood of
being insolvent. Comparing with the results reported in column (II), the%
\textbf{\ }incorporation of network effects in the Probit model does not
seem to change significantly the estimated coefficients for this data set,
although the standard errors are slightly smaller for the proposed method,
which results in the age risk variable being significant in the mixed
graphical approach but not in the conventional approach.

Table \ref{tabperf} reports the classification accuracy statistics on the
hold-out sample, for the mixed graphical Probit, the mixed Probit with
uncorrelated random effects and the conventional Probit, that ignores
unobserved heterogeneity and is often used in credit risk modelling. When
adopting the mixed graphical Probit, the overall classification accuracy is
significantly improved. Given the high number of non-failed companies in the
data, the mixed graphical Probit is particularly good at identifying
correctly companies that did not fail. This is confirmed also by the ROC
curve in Figure \ref{roc_app}, where the ROC of the mixed graphical Probit
lies always above the ROC of the mixed Probit and conventional Probit.

\begin{figure}
\centering
{\includegraphics[width=0.5\textwidth]{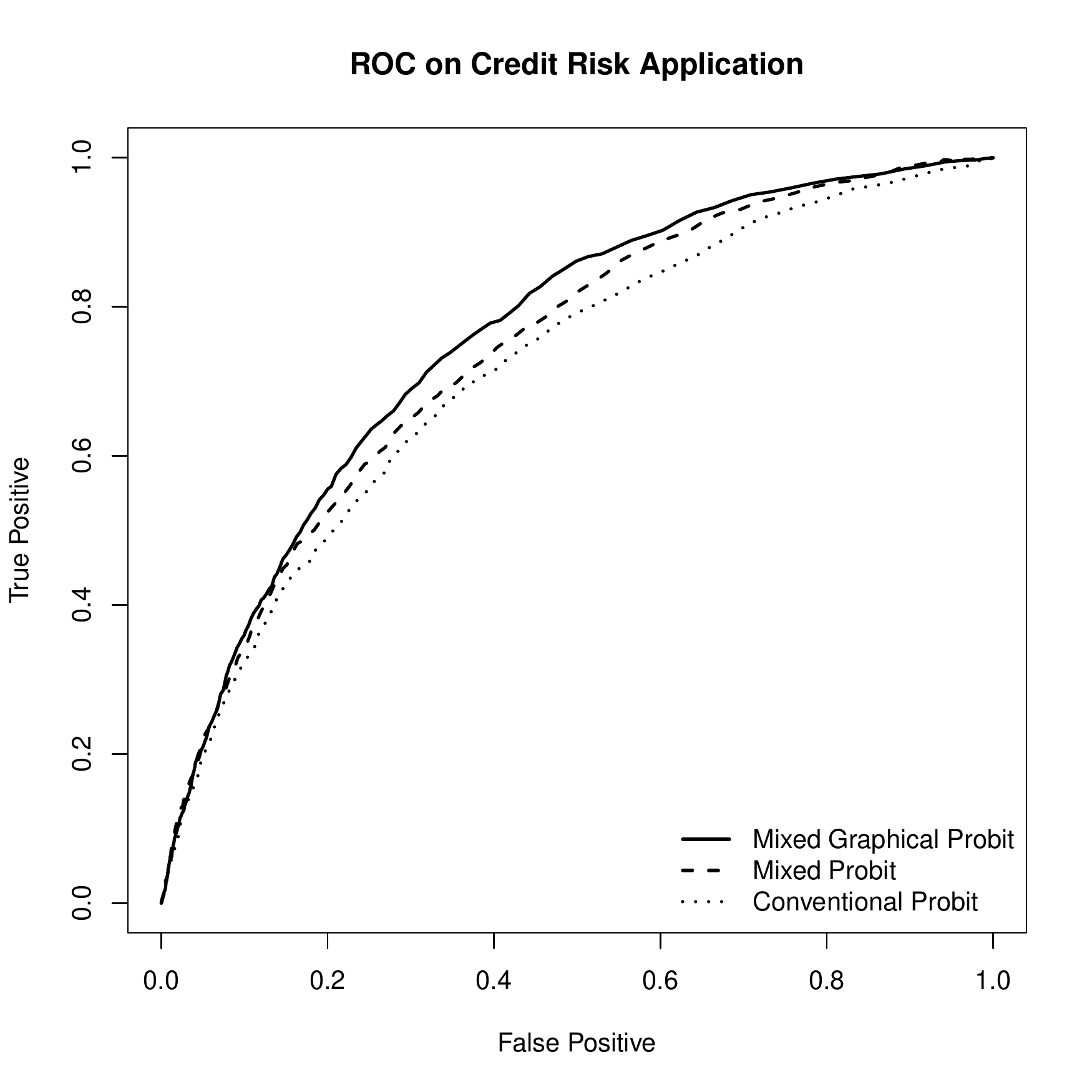} }
\caption{\label{roc_app} ROC curves of predicted outcomes
on test sample: comparison between mixed graphical Probit, mixed Probit
and conventional Probit on the credit risk application.}
\end{figure}

To explore the network of dependencies we have also performed penalised
likelihood estimation, selecting the penalisation parameter as that yielding
the highest percentage of non-zero and zero outcomes correctly predicted.
Results are very similar to those from maximum likelihood in terms of
estimation of regression coefficients. However, one interesting output of
penalised likelihood estimation is the estimated sparse precision matrix,
which gives an indication of the more connected sectors in the economy.
Figure \ref{fig_app} shows the estimated network, where links between any
two sectors appear when there is a non-zero conditional correlation among
them. It is interesting to see that the sectors that are more central to the
network are those from the real estate, manufacturing industry, and the
activities of households as employers, whereas we mostly find services
activities sectors, and in particular, the sectors ``arts, entertainment and
recreation'' and ``transportation and storage'' not highly connected.

\begin{table}
\caption{Regression coefficients and standard errors estimated by the
proposed credit risk model and a conventional mixed probit model on the
training sample of the credit risk application. (*) denotes significance at the 5\% level. \label{tabcr}}
\centering
\begin{tabular}{|l|cc|cc|}
\hline
& \multicolumn{2}{|c|}{(I): {\footnotesize Mixed Graphical Probit}} & \multicolumn{2}{|c|}{
(II): {\footnotesize Mixed Probit}} \\ \hline
{\footnotesize Variable} & {\footnotesize Parameter} & {\footnotesize Standard Error}
& {\footnotesize Parameter} & {\footnotesize Standard Error} \\ \hline
{\footnotesize Total liabilities/total assets} & {\footnotesize 0.0215} & {\footnotesize 0.0146} &
{\footnotesize 0.0218} & {\footnotesize 0.0150} \\
{\footnotesize Networth/total liabilities} &
{\footnotesize 0.0001} & {\footnotesize 0.0011} &
{\footnotesize 0.0001} & {\footnotesize 0.0013} \\
{\footnotesize Cash/total assets} & {\footnotesize %
-0.1150*} & {\footnotesize 0.0359} & {\footnotesize %
-0.1157*} & {\footnotesize 0.0409} \\
{\footnotesize Current liabilities/current assets} &
{\footnotesize -0.0063} & {\footnotesize 0.0088} & %
{\footnotesize -0.0063} & {\footnotesize 0.0091} \\
{\footnotesize Retained profits/total assets} &
{\footnotesize -0.1428*} & {\footnotesize 0.0248} &
{\footnotesize -0.1428*} & {\footnotesize 0.0250} \\
{\footnotesize Account receivable/total liabilities} &
{\footnotesize -0.7437} & {\footnotesize 1.3177} &
{\footnotesize -0.7378} & {\footnotesize 1.4844} \\
{\footnotesize Trade credit/total liabilities} &
{\footnotesize 0.0309} & {\footnotesize 0.0345} &
{\footnotesize 0.0354} & {\footnotesize 0.0386} \\
{\footnotesize Trade debt/total assets} & {\footnotesize %
-0.2094*} & {\footnotesize 0.0507} & {\footnotesize %
-0.2084*} & {\footnotesize 0.0553} \\
{\footnotesize Size} & {\footnotesize -0.0814*} &
{\footnotesize 0.0039} & {\footnotesize -0.0807*} &
{\footnotesize 0.0053} \\
{\footnotesize Age} & {\footnotesize -0.1956*} &
{\footnotesize 0.0116} & {\footnotesize -0.1966}$^{\ast
} $ & {\footnotesize 0.0144} \\
{\footnotesize Age risk} & {\footnotesize 0.0479*} &
{\footnotesize 0.0246} & {\footnotesize 0.0485} &
{\footnotesize 0.0252} \\
{\footnotesize Regional GDP} & {\footnotesize 0.0133} &
{\footnotesize 0.0101} &{\footnotesize 0.0131} &
{\footnotesize 0.0242} \\ \hline
\end{tabular}
\end{table}

\begin{table}
\caption{Performance of the credit risk models on the testing sample.\label{tabperf}}
\centering
\begin{tabular}{|l|c|c|}
\hline
& \multicolumn{2}{|c|}{\footnotesize Percentage correctly classified} \\
\hline
& \multicolumn{1}{|c|}{\footnotesize Non-failed} & {\footnotesize Failed} \\
\hline
{\footnotesize Mixed Graphical Probit} & {\footnotesize 66.43} &
{\footnotesize 73.03} \\
{\footnotesize Mixed Probit} & {\footnotesize 64.67} & {\footnotesize 69.70}
\\
{\footnotesize Conventional Probit} & {\footnotesize 62.99} & {\footnotesize %
66.54} \\ \hline
\end{tabular}
\end{table}

\begin{figure}
\centering
{\includegraphics[width=\textwidth]{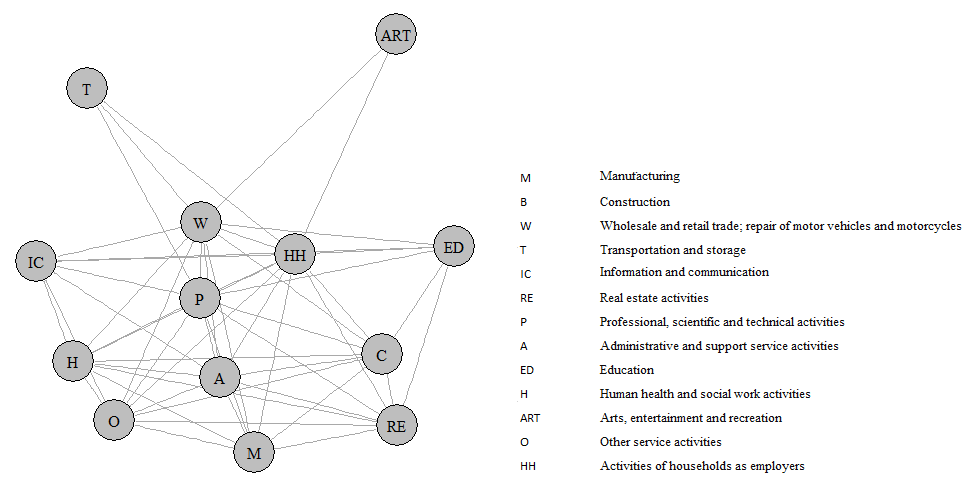} }
\caption{\label{fig_app} Credit risk application: estimated network between
sectors of economic activity.}
\end{figure}

\section{Concluding remarks} \label{concl}

In this paper we have proposed a computationally efficient EM\ algorithm for
ML estimation of a mixed Probit model with correlated group-specific
effects and have shown its use in a credit risk application, for which existing approaches were prohibitively slow.
We have proposed unconstrained and penalised
likelihood estimation approaches for inference and have derived the observed information matrix and asymptotic
standard errors of the estimates. The penalised approach is
suitable for when the number of groups is large relative to the number of
observations, for which maximum likelihood fails, or when the underlying
network is expected to be sparse. An extensive simulation study showed that our
proposed estimator has good finite sample properties and can be adopted for
estimation and prediction using very large data sets, given its moderate
computational costs.

A large-scale credit risk application on a unique dataset on SMEs, a setting in which credit risk modelling is currently under-developed, showed that accounting for network effects makes a significant contribution to increasing the default prediction power of risk
models and therefore that efficient inferential procedures for these
models are particularly useful in this field.

\section*{Acknowledgements}
The authors acknowledge financial support from EPSRC [EP/L021250/1]. We thank the financial institution that provided the data, George Foy for
assisting with data retrieval, and Francesco Moscone, Sergio di Cesare and
Mark Lycett for helpful comments on this manuscript.

\bibliographystyle{chicago}	
\bibliography{scribesept16}

\section*{Appendix A: Moments of truncated normals} \label{truncnorm}
       \renewcommand{\theequation}{A.\arabic{equation}}
       \renewcommand{\thesection}{A}
       \setcounter{equation}{0}

       \medskip

We now provide the formulae for deriving the central and non-central moments
of $y_{ir}^{\ast }$ given $\mathbf{y}_{-i,r}^{\ast },y_{ir}$. By the theorem
on conditional normals, $y_{ir}^{\ast }$ given $\mathbf{y}_{-i,r}^{\ast }$
has a normal distribution with mean and variance:%
\begin{eqnarray*}
\tilde{\mu}_{ir} &=&\bm{\beta }^{\prime }\mathbf{x}_{ir}+\mathbf{\Sigma }%
_{r,i,-i}\mathbf{\Sigma }_{r,-i,-i}^{-1}\left( \mathbf{y}_{-i,r}^{\ast }-%
\mathbf{X}_{-i,r}\bm{\beta }\right) , \\
\tilde{\sigma}_{ir}^{2} &=&\sigma _{ir}^{2}-\mathbf{\Sigma }_{r,i,-i}\mathbf{%
\Sigma }_{r,-i,-i}^{-1}\mathbf{\Sigma }_{r,-i,i},
\end{eqnarray*}%
where $\sigma _{ir}^{2}$ is the ($i,i$)th element of $\mathbf{\Sigma }_{r}$.
It follows that the conditional distribution of $y_{ir}^{\ast }$ given $%
\mathbf{y}_{-i,r}^{\ast },y_{ir}$ is a truncated normal. Let $\xi _{ir,1}=%
\frac{t_{1}-\tilde{\mu}_{ir}}{\tilde{\sigma}_{ir}}$, $\xi _{ir,2}=\frac{%
t_{2}-\tilde{\mu}_{ir}}{\tilde{\sigma}_{ir}}$ and $\rho _{1,ir}=\frac{\phi
\left( \xi _{ir,1}\right) -\phi \left( \xi _{ir,2}\right) }{\Phi \left( \xi
_{ir,2}\right) -\Phi \left( \xi _{ir,1}\right) }$, $\rho _{2,ir}=\frac{\xi
_{ir,2}\phi \left( \xi _{ir,1}\right) -\xi _{ir,1}\phi \left( \xi
_{ir,2}\right) }{\Phi \left( \xi _{ir,2}\right) -\Phi \left( \xi
_{ir,1}\right) }$ with
\begin{equation*}
t_{1}=\left\{
\begin{array}{c}
0,\text{ if }y_{ir}=1 \\
-\infty ,\text{ if }y_{ir}=0%
\end{array}%
\right. ,t_{2}=\left\{
\begin{array}{c}
\infty ,\text{ if }y_{ir}=1 \\
0,\text{ if }y_{ir}=0%
\end{array}%
\right. ,
\end{equation*}%
and$\ \phi $, $\Phi $ are the density and cumulative distribution,
respectively, of a standard normal distribution. The first and second
moments of $y_{ir}^{\ast }$ given $\mathbf{y}_{-i,r}^{\ast },y_{ir}$ are:%
\begin{eqnarray*}
\lambda _{i,1} &=&\tilde{\mu}_{ir}+\rho _{1,ir}\tilde{\sigma}_{ir},
\label{meant} \\
\lambda _{i,2} &=&\tilde{\mu}_{ir}^{2}+\tilde{\sigma}_{ir}^{2}+2\rho _{1,ir}%
\tilde{\sigma}_{ir}\tilde{\mu}_{ir}+\rho _{2,ir}\tilde{\sigma}_{ir}^{2},
\label{vart}
\end{eqnarray*}%
while the second, third and fourth central moments of $y_{ir}^{\ast }$ given
$\mathbf{y}_{-i,r}^{\ast },y_{ir}$ are (see \citealp{Horrace2015}):%
\begin{eqnarray*}
\lambda _{i,2}^{c} &=&\tilde{\sigma}_{ir}^{2}-\tilde{\sigma}_{ir}\rho
_{1,ir}\lambda _{i,1} \\
\lambda _{i,3}^{c} &=&\tilde{\sigma}_{ir}\rho _{1,ir}\left( \lambda
_{i,1}^{2}-\lambda _{i,2}^{c}\right) , \\
\lambda _{i,4}^{c} &=&2\tilde{\sigma}_{ir}^{4}-3\left( \tilde{\sigma}%
_{ir}\rho _{1,ir}\lambda _{i,1}^{c}\right) ^{2}-\tilde{\sigma}_{ir}^{-1}\rho
_{1,ir}\lambda _{i,3}^{c}+\tilde{\mu}_{ir}^{2}\lambda _{i,2}^{c}.
\end{eqnarray*}

\renewcommand{\thesection}{2}

\section*{Appendix B: Conditional expectations} \label{app2}
       \renewcommand{\theequation}{B.\arabic{equation}}
       \renewcommand{\thesection}{B}
       \setcounter{equation}{0}

       \medskip

Using the law of iterated expectations we know that:%
\begin{eqnarray*}
E\left( \mathbf{u}_{r}\mathbf{|y}_{r}\right) &=&E\left[ E\left( \mathbf{u}%
_{r}\mathbf{|y}_{r}^{\ast }\right) \mathbf{|y}_{r}\right] , \\
E\left( \mathbf{u}_{r}\mathbf{u}_{r}^{\prime }\mathbf{|y}_{r}\right) &=&E%
\left[ E\left( \mathbf{u}_{r}\mathbf{u}_{r}^{\prime }\mathbf{|y}_{r}^{\ast
}\right) \mathbf{|y}_{r}\right] .
\end{eqnarray*}%
Noting that%
\begin{equation*}
\left(
\begin{array}{c}
\mathbf{u}_{r} \\
\mathbf{y}_{r}^{\ast }%
\end{array}%
\right) \sim N\left(
\begin{array}{c}
\mathbf{0} \\
\mathbf{X}_{r}\bm{\beta }%
\end{array}%
,%
\begin{array}{cc}
\mathbf{\Sigma }_{G} & \mathbf{\Sigma }_{G}\mathbf{Z}_{r}^{\prime } \\
\mathbf{Z}_{r}\mathbf{\Sigma }_{G} & \mathbf{\Sigma }_{r}%
\end{array}%
\right) ,
\end{equation*}%
we can use the theorem on conditional normals to obtain:%
\begin{equation*}
E\left( \mathbf{u}_{r}\mathbf{|y}_{r}^{\ast }\right) =\mathbf{\Sigma }_{G}%
\mathbf{Z}_{r}^{\prime }\mathbf{\Sigma }_{r}^{-1}\left( \mathbf{y}_{r}^{\ast
}-\mathbf{X}_{r}\bm{\beta }\right) ,
\end{equation*}%
so that (\ref{Eu1}) holds. Similarly, focusing on $E\left( \mathbf{u}_{r}%
\mathbf{u}_{r}^{\prime }\mathbf{|y}_{r}\right) $ and using again the theorem
on conditional normals we obtain (\ref{euugiveny}).

\renewcommand{\thesection}{2}

\section*{Appendix C: Useful results on moments of quadratic forms} \label{app1}
       \renewcommand{\theequation}{C.\arabic{equation}}
       \renewcommand{\thesection}{C}
       \setcounter{equation}{0}

       \medskip

In the following we provide a set of results that are useful for our
theoretical derivations in Appendix D.

\begin{lemma}
\label{lemma1}Let $\mathbf{z}\sim N\left( \bm{\mu },\mathbf{\Sigma }%
\right) $ be a $n$-dimensional random vector, and $\mathbf{A,B}$ two $%
n\times n$ symmetric matrices. Then:%
\begin{eqnarray*}
E\left( \mathbf{z}^{\prime }\mathbf{Az}\right) &=&Tr\left( \mathbf{\Sigma A}%
\right) +\bm{\mu }^{\prime }\bm{A\mu }, \\
E\left( \mathbf{z}^{\prime }\mathbf{Azz}^{\prime }\mathbf{Bz}\right) &=&%
\left[ Tr(\mathbf{A\Sigma })+\bm{\mu }^{\prime }\bm{A\mu }\right] %
\left[ Tr(\mathbf{B\Sigma })+\bm{\mu }^{\prime }\bm{B\mu }\right]
+2Tr\left( \mathbf{\Sigma \mathbf{A}\Sigma B}\right) +4\bm{\mu }^{\prime
}\bm{A\Sigma B\mu } \\
E\left[ \left( \mathbf{z-x}\right) ^{\prime }\mathbf{A}\left( \mathbf{z-x}%
\right) \right] &=&Tr\left( \mathbf{\Sigma A}\right) +\left( \bm{\mu -x}%
\right) ^{\prime }\mathbf{A}\left( \bm{\mu -x}\right) ,
\end{eqnarray*}
\end{lemma}

\begin{proof}
See \citet{Ullah2004}.
\end{proof}

\begin{lemma}
\label{lemma2}Let $\mathbf{z}$ be a $n$-dimensional vector of non-normal
random variables with mean $\bm{\mu }$ and covariance $\mathbf{\Sigma}
=diag(\sigma _{1}^{2},...,\sigma _{n}^{2})$, let $\mathbf{x}$ be a $n$%
-dimensional non-random vector, $\mathbf{A}$, $\mathbf{B}$ two $n\times n$
symmetric matrices, and:%
\begin{eqnarray*}
\mathbf{\Lambda }_{3} &=&diag\left\{ E\left( z_{1}-\mu _{1}\right)
^{3}/\left( \sigma _{1}^{2}\right) ^{3/2},...,E\left( z_{n}-\mu _{n}\right)
^{3}/\left( \sigma _{n}^{2}\right) ^{3/2}\right\} , \\
\mathbf{\Lambda }_{4} &=& diag\left\{ E\left( z_{1}-\mu _{1}\right)
^{4}/\left( \sigma _{1}^{2}\right) ^{2}-3,...,E\left( z_{n}-\mu _{n}\right)
^{4}/\left( \sigma _{n}^{2}\right) ^{2}-3\right\} .
\end{eqnarray*}%
Then:%
\begin{eqnarray}
E\left[ \left( \bm{z-\mu }\right) ^{\prime }\bm{A\left( \bm{%
z-\mu }\right) }\left( \bm{z-\mu }\right) ^{\prime }\mathbf{B}\left(
\bm{z-\mu }\right) \right] &=&Tr\left( \mathbf{\Sigma }^{1/2}\mathbf{%
B\Sigma }^{1/2}\mathbf{\Lambda }_{4}diag\left( \mathbf{\mathbf{\Sigma }}%
^{1/2}\mathbf{A\Sigma }^{1/2}\right) \right)  \notag \\
&+&Tr(\mathbf{A\Sigma })Tr(\mathbf{B\Sigma })+2Tr\left( \mathbf{\mathbf{%
\Sigma }B\mathbf{\Sigma }A}\right)  \label{ullah3} \\
E\left[ \left( \bm{z-\mu }\right) ^{\prime }\mathbf{A\left( \bm{%
z-\mu }\right) }\left( \bm{z-\mu }\right) \right] &=&diag(\mathbf{%
\Lambda }_{3}\mathbf{\Sigma }^{1/2}\mathbf{A\mathbf{\Sigma }}^{1/2}\mathbf{)%
\mathbf{\Sigma }}^{1/2}\mathbf{1}_{n}.  \label{ullah4}
\end{eqnarray}
\end{lemma}

\begin{proof}
See \citet{Wiens1992} and \citet{Ullah2004}.
\end{proof}

We observe that in the case in which $\mathbf{A}$, $\mathbf{B}$ are
asymmetric, results in the above two Lemmas still hold with $\mathbf{\left(
\mathbf{A+A}^{\prime }\right) }/2$ and $\left( \mathbf{B+B}^{\prime }\right)
/2$ in place of $\mathbf{A}$ and $\mathbf{B}$.

\begin{corollary}
\label{cor1}Under the conditions of Lemma \ref{lemma2}, let $\mathbf{x}$ be
a $n$-dimensional non-random vector. Then:%
\begin{eqnarray}
&&E\left[ \left( \mathbf{z-x}\right) ^{\prime }\mathbf{A\left( \mathbf{z-x}%
\right) }\cdot \left( \mathbf{z-x}\right) ^{\prime }\mathbf{B\left( \mathbf{%
z-x}\right) }\right]  \notag  \label{ex0} \\
&=&Tr\left[ \mathbf{\Sigma }^{1/2}\mathbf{B\Sigma }^{1/2}\mathbf{\Lambda }%
_{4}diag\left( \mathbf{\mathbf{\Sigma }}^{1/2}\mathbf{A\Sigma }^{1/2}\right) %
\right] +Tr(\mathbf{A\Sigma })Tr(\mathbf{B\Sigma })+2Tr\left( \mathbf{%
\mathbf{\Sigma }B\mathbf{\Sigma }A}\right)  \notag  \label{ex1} \\
&+&2\left[ diag(\mathbf{\Lambda }_{3}\mathbf{\Sigma }^{1/2}\mathbf{A\mathbf{%
\Sigma }}^{1/2}\mathbf{)1}_{n}\right] ^{\prime }\mathbf{\Sigma }^{1/2}%
\mathbf{B}\left( \bm{\mu -x}\right)  \notag  \label{ex2} \\
&+&2\left[ diag(\mathbf{\Lambda }_{3}\mathbf{\Sigma }^{1/2}\mathbf{B\mathbf{%
\Sigma }}^{1/2}\mathbf{)1}_{n}\right] ^{\prime }\mathbf{\Sigma }^{1/2}%
\mathbf{A}\left( \bm{\mu -x}\right)  \notag  \label{ex3} \\
&+&Tr\left( \mathbf{\Sigma A}\right) \left( \bm{\mu -x}\right) ^{\prime }%
\mathbf{B}\left( \bm{\mu -x}\right) +Tr\left( \mathbf{\Sigma B}\right)
\left( \bm{\mu -x}\right) ^{\prime }\mathbf{A}\left( \bm{\mu -x}%
\right)  \notag \\
&+&4\left( \bm{\mu -x}\right) ^{\prime }\mathbf{A}\mathbf{\Sigma} \mathbf{B}%
\left( \bm{\mu -x}\right) +\left( \bm{\mu -x}\right) ^{\prime }%
\mathbf{A}\left( \bm{\mu -x}\right) \left( \bm{\mu -x}\right)
^{\prime }\mathbf{B}\left( \bm{\mu -x}\right) .  \label{ex5}
\end{eqnarray}%
and%
\begin{eqnarray}
E\left[ \left( \mathbf{z-x}\right) ^{\prime }\mathbf{A\left( \mathbf{z-x}%
\right) }\left( \mathbf{z-x}\right) \right] &=&diag(\mathbf{\Lambda }_{3}%
\mathbf{\Sigma }^{1/2}\mathbf{A\mathbf{\Sigma }}^{1/2}\mathbf{)\mathbf{%
\Sigma }}^{1/2}\mathbf{1}_{n}+Tr\left( \mathbf{A\mathbf{\Sigma }}\right)
\left( \bm{\mu -x}\right)  \notag  \label{ex6} \\
&+&2\mathbf{\mathbf{\Sigma }A}\left( \bm{\mu -x}\right) +\left( \bm{%
\mu -x}\right) ^{\prime }\mathbf{A}\left( \bm{\mu -x}\right) \left(
\bm{\mu -x}\right).  \label{ex7}
\end{eqnarray}
\end{corollary}

\begin{proof}
To show the above results, first note that we can write:%
\begin{equation*}
\left( \mathbf{z-x}\right) ^{\prime }\mathbf{A\left( \mathbf{z-x}\right) }%
=\left( \bm{z-\mu }\right) ^{\prime }\mathbf{A}\left( \bm{z-\mu }%
\right) +2\left( \bm{z-\mu }\right) ^{\prime }\mathbf{A}\left( \bm{%
\mu -x}\right) +\left( \bm{\mu -x}\right) ^{\prime }\mathbf{A}\left(
\bm{\mu -x}\right).
\end{equation*}%
Hence:%
\begin{eqnarray*}
&&\left( \mathbf{z-x}\right) ^{\prime }\mathbf{A\left( \mathbf{z-x}\right) }%
\left( \mathbf{z-x}\right) ^{\prime }\mathbf{B\left( \mathbf{z-x}\right) } \\
&=&\left( \bm{z-\mu }\right) ^{\prime }\mathbf{A}\left( \bm{z-\mu }%
\right) \left[ \left( \bm{z-\mu }\right) ^{\prime }\mathbf{B}\left(
\bm{z-\mu }\right) +2\left( \bm{z-\mu }\right) ^{\prime }\mathbf{B}%
\left( \bm{\mu -x}\right) +\left( \bm{\mu -x}\right) ^{\prime }%
\mathbf{B}\left( \bm{\mu -x}\right) \right] \\
&+&2\left( \bm{z-\mu }\right) ^{\prime }\mathbf{A}\left( \bm{\mu -x}%
\right) \left[ \left( \bm{z-\mu }\right) ^{\prime }\mathbf{B}\left(
\bm{z-\mu }\right) +2\left( \bm{z-\mu }\right) ^{\prime }\mathbf{B}%
\left( \bm{\mu -x}\right) +\left( \bm{\mu -x}\right) ^{\prime }%
\mathbf{B}\left( \bm{\mu -x}\right) \right] \\
&+&\left( \bm{\mu -x}\right) ^{\prime }\mathbf{A}\left( \bm{\mu -x}%
\right) \left[ \left( \bm{z-\mu }\right) ^{\prime }\mathbf{B}\left(
\bm{z-\mu }\right) +2\left( \bm{z-\mu }\right) ^{\prime }\mathbf{B}%
\left( \bm{\mu -x}\right) +\left( \bm{\mu -x}\right) ^{\prime }%
\mathbf{B}\left( \bm{\mu -x}\right) \right].
\end{eqnarray*}%
So that we obtain%
\begin{eqnarray*}
&&E\left[ \left( \mathbf{z-x}\right) ^{\prime }\mathbf{A\left( \mathbf{z-x}%
\right) }\left( \mathbf{z-x}\right) ^{\prime }\mathbf{B\left( \mathbf{z-x}%
\right) }\right] \\
&=&E\left[ \left( \bm{z-\mu }\right) ^{\prime }\mathbf{A}\left( \bm{%
z-\mu }\right) \left( \bm{z-\mu }\right) ^{\prime }\mathbf{B}\left(
\bm{z-\mu }\right) \right] +2E\left[ \left( \bm{z-\mu }\right)
^{\prime }\mathbf{A}\left( \bm{z-\mu }\right) \left( \bm{z-\mu }%
\right) ^{\prime }\mathbf{B}\left( \bm{\mu -x}\right) \right] \\
&+&Tr\left( \mathbf{\Sigma A}\right) \left( \bm{\mu -x}\right) ^{\prime }%
\mathbf{B}\left( \bm{\mu -x}\right) +2E\left[ \left( \bm{z-\mu }%
\right) ^{\prime }\mathbf{A}\left( \bm{\mu -x}\right) \left( \bm{%
z-\mu }\right) ^{\prime }\mathbf{B}\left( \bm{z-\mu }\right) \right] \\
&+&4\left( \bm{\mu -x}\right) ^{\prime }\mathbf{A\Sigma B}\left( \bm{%
\mu -x}\right) +\left( \bm{\mu -x}\right) ^{\prime }\mathbf{A}\left(
\bm{\mu -x}\right) Tr\left( \mathbf{\Sigma B}\right) \\
&+&\left( \bm{\mu
-x}\right) ^{\prime }\mathbf{A}\left( \bm{\mu -x}\right) \left( \bm{%
\mu -x}\right) ^{\prime }\mathbf{B}\left( \bm{\mu -x}\right).
\end{eqnarray*}%
Simplifying and substituting results (\ref{ullah3}) and (\ref{ullah4}) in
the above we obtain (\ref{ex5}). Similarly, we have for (\ref{ex7}):%
\begin{eqnarray*}
&&\left( \mathbf{z-x}\right) ^{\prime }\mathbf{A\left( \mathbf{z-x}\right) B}%
\left( \mathbf{z-x}\right) \\
&=&\left( \bm{z-\mu }\right) ^{\prime }\mathbf{A}\left( \bm{z-\mu }%
\right) \mathbf{B}\left( \bm{z-\mu }\right) +\left( \bm{z-\mu }%
\right) ^{\prime }\mathbf{A}\left( \bm{z-\mu }\right) \mathbf{B}\left(
\bm{\mu -x}\right) \\
&&+2\left( \bm{z-\mu }\right) ^{\prime }\mathbf{A}\left( \bm{\mu -x}%
\right) \mathbf{B}\left( \bm{z-\mu }\right) +2\left( \bm{z-\mu }%
\right) ^{\prime }\mathbf{A}\left( \bm{\mu -x}\right) \mathbf{B}\left(
\bm{\mu -x}\right) \\
&&+\left( \bm{\mu -x}\right) ^{\prime }\mathbf{A}\left( \bm{\mu -x}%
\right) \mathbf{B}\left( \bm{z-\mu }\right) +\left( \bm{\mu -x}%
\right) ^{\prime }\mathbf{A}\left( \bm{\mu -x}\right) \mathbf{B}\left(
\bm{\mu -x}\right),
\end{eqnarray*}%
leading to%
\begin{eqnarray*}
&&E\left[ \left( \mathbf{z-x}\right) ^{\prime }\mathbf{A\left( \mathbf{z-x}%
\right) B}\left( \mathbf{z-x}\right) \right] \\
&=&E\left[ \left( \bm{z-\mu }\right) ^{\prime }\mathbf{A}\left( \bm{%
z-\mu }\right) \mathbf{B}\left( \bm{z-\mu }\right) \right] +Tr\left(
\mathbf{A\mathbf{\Sigma }}\right) \mathbf{B}\left( \bm{\mu -x}\right) +2%
\mathbf{B\mathbf{\Sigma }A}\left( \bm{\mu -x}\right) \\
&&+\left( \bm{\mu -x}\right) ^{\prime }\mathbf{A}\left( \bm{\mu -x}%
\right) \mathbf{B}\left( \bm{\mu -x}\right) .
\end{eqnarray*}%
Noting that $E\left[ \left( \bm{z-\mu }\right) ^{\prime }\mathbf{A\left(
\bm{z-\mu }\right) }\left( \bm{z-\mu }\right) \right] =diag(\mathbf{%
\Lambda }_{3}\mathbf{\Sigma }^{1/2}\mathbf{A\mathbf{\Sigma }}^{1/2}\mathbf{)%
\mathbf{\Sigma }}^{1/2}\mathbf{1}_{n}$, we obtain (\ref{ex7}).
\end{proof}

\renewcommand{\thesection}{2}

\section*{Appendix D: Asymptotic standard errors of ML estimates} \label{app3}
       \renewcommand{\theequation}{D.\arabic{equation}}
       \renewcommand{\thesection}{D}
       \setcounter{equation}{0}

       \medskip

Let $\mathbf{X,A,B}$ be square $n\times n$ matrices, with $\mathbf{X}$
nonsingular, $\mathbf{a}$, $\mathbf{b}$ two $n$-dimensional vectors, and $%
f\left( .\right) $ a scalar function. In the rest of this section we use the
following results on derivatives (\citealp{Bernstein2005})%
\begin{eqnarray*}
\frac{\partial \log \left\vert \mathbf{X}\right\vert }{\partial \mathbf{X}}
&=&\mathbf{X}^{-1}, \\
\frac{\partial \mathbf{a}^{\prime }\mathbf{Xb}}{\partial \mathbf{X}} &=&%
\mathbf{ab}^{\prime }, \\
\frac{\partial Tr\left( \mathbf{AX}^{-1}\mathbf{B}\right) }{\partial \mathbf{%
X}} &=&-\mathbf{X}^{-1}\mathbf{BAX}^{-1}, \\
\frac{df\left( \mathbf{X}\right) }{dx_{ij}} &=&Tr\left[ \left( \frac{%
\partial f}{\partial \mathbf{X}}\right) ^{\prime }\frac{\partial \mathbf{X}}{%
\partial x_{ij}}\right] .
\end{eqnarray*}%
In addition, we will use the following notation:%
\begin{eqnarray*}
\mathbf{\Sigma }_{r} &=&{\rm Var}\left[ \left( \mathbf{y}_{r}^{\ast }-\mathbf{X}%
_{r}\bm{\beta }\right) \left( \mathbf{y}_{r}^{\ast }-\mathbf{X}_{r}%
\bm{\beta }\right) ^{\prime }|\mathbf{y}_{r}\right] \\
\mathbf{A}^{gh} &=&\mathbf{Z}_{r}\mathbf{M_{r}}^{-1}\mathbf{J}^{gh}\mathbf{M_{r}}^{-1}\mathbf{Z%
}_{r}^{\prime },
\end{eqnarray*}%
where $\mathbf{J}^{gh}$ is a $G\times G$ matrix of zeros with 1 on the ($g,h$%
) and ($h,g$) positions and $\mathbf{M_{r}}={\rm diag}(m_{1r},m_{2r},...,m_{Gr})$.

Consider:%
\begin{eqnarray*}
\log f_{\mathbf{y},\mathbf{y}^{\ast },\mathbf{u}}\left( \mathbf{y},\mathbf{y}%
^{\ast },\mathbf{u}|\bm{\vartheta }\right) &\approx &\frac{R}{2}\log
\left\vert \mathbf{\Phi }_{G}\right\vert -\frac{1}{2}\sum_{r=1}^{R}\mathbf{u}%
_{r}^{\prime }\mathbf{\Phi }_{G}\mathbf{u}_{r} \\
&&-\frac{1}{2}\sum_{r=1}^{R}\left( \mathbf{y}_{r}^{\ast }-\mathbf{X}_{r}%
\bm{\beta }-\mathbf{Z}_{r}\mathbf{u}_{r}\right) ^{\prime }\left( \mathbf{%
y}_{r}^{\ast }-\mathbf{X}_{r}\bm{\beta }-\mathbf{Z}_{r}\mathbf{u}%
_{r}\right) .
\end{eqnarray*}%
For ease of exposition, in the following we write $\log f_{\mathbf{y},%
\mathbf{y}^{\ast },\mathbf{u}}\left( \mathbf{y},\mathbf{y}^{\ast },\mathbf{u}%
|\bm{\vartheta }\right) $ as $\log f_{\mathbf{y},\mathbf{y}^{\ast },%
\mathbf{u}}$. The first and second derivatives of $\log f_{\mathbf{y},%
\mathbf{y}^{\ast },\mathbf{u}}$ with respect to $\bm{\vartheta }$ are:%
\begin{eqnarray*}
\frac{\partial \log f_{\mathbf{y},\mathbf{y}^{\ast },\mathbf{u}}}{\partial
\bm{\beta }} &=&\sum_{r=1}^{R}\mathbf{X}_{r}^{\prime }\left( \mathbf{y}%
_{r}^{\ast }-\mathbf{X}_{r}\bm{\beta }-\mathbf{Z}_{r}\mathbf{u}%
_{r}\right) , \\
\frac{\partial \log f_{\mathbf{y},\mathbf{y}^{\ast },\mathbf{u}}}{\partial
\phi _{gh}} &=&\frac{R}{2}Tr\left( \mathbf{\Phi }_{G}^{-1}\mathbf{J}%
^{gh}\right) -\sum_{r=1}^{R}u_{gr}u_{hr}, \\
\frac{\partial ^{2}\log f_{\mathbf{y},\mathbf{y}^{\ast },\mathbf{u}}}{%
\partial\beta \partial \beta'} &=&-\sum_{r=1}^{R}\mathbf{%
X}_{r}^{\prime }\mathbf{X}_{r}, \\
\frac{\partial \log f_{\mathbf{y},\mathbf{y}^{\ast },\mathbf{u}}}{\partial
\phi _{gh}\partial \phi _{k\ell }} &=&-\frac{R}{2}Tr\left( \mathbf{\Phi }%
_{G}^{-1}\mathbf{J}^{gh}\mathbf{\Phi }_{G}^{-1}\mathbf{J}^{k\ell }\right) ,
\\
\frac{\partial ^{2}\log f_{\mathbf{y},\mathbf{y}^{\ast },\mathbf{u}}}{%
\partial \bm{\beta }\partial \phi _{gh}} &=&\mathbf{0.}
\end{eqnarray*}%
Standard errors of $\widehat{\bm{\beta }}$ can be obtained by
substituting the formulas above in (\ref{B}). In particular, let $B\left(
\mathbf{y}|\bm{\vartheta }\right) $ in (\ref{B}) be structured as
follows:%
\begin{equation*}
B\left( \mathbf{y}|\bm{\vartheta }\right) =\left(
\begin{tabular}{l|lll}
$B_{\bm{\beta \beta }}$ & $B_{\bm{\beta }\phi _{11}}$ & $...$ & $B_{%
\bm{\beta }\phi _{GG}}$ \\ \hline
$B_{\bm{\beta }\phi _{11}}$ & $B_{\phi _{11}\phi _{11}}$ & $...$ & $%
B_{\phi _{11}\phi _{GG}}$ \\
$...$ & $...$ & $...$ & $...$ \\
$B_{\bm{\beta }\phi _{GG}}$ & $B_{\phi _{11}\phi _{GG}}$ & $...$ & $%
B_{\phi _{GG}\phi _{GG}}$%
\end{tabular}%
\right) .
\end{equation*}%
It is easy to see that%
\begin{eqnarray*}
B_{\bm{\beta \beta }} &=&-\sum_{r=1}^{R}\mathbf{X}_{r}^{\prime }\mathbf{X%
}_{r}+\sum_{r=1}^{R}\mathbf{X}_{r}^{\prime }E\left[ \left( \mathbf{y}%
_{r}^{\ast }-\mathbf{X}_{r}\bm{\beta }-\mathbf{Z}_{r}\mathbf{u}%
_{r}\right) \left( \mathbf{y}_{r}^{\ast }-\mathbf{X}_{r}\bm{\beta }-%
\mathbf{Z}_{r}\mathbf{u}_{r}\right) ^{\prime }|\mathbf{y}_{r}\right] \mathbf{%
X}_{r} \\
&&-\sum_{r=1}^{R}\mathbf{X}_{r}^{\prime }E\left[ \left( \mathbf{y}_{r}^{\ast
}-\mathbf{X}_{r}\bm{\beta }-\mathbf{Z}_{r}\mathbf{u}_{r}\right) |\mathbf{%
y}_{r}\right] E\left[ \left( \mathbf{y}_{r}^{\ast }-\mathbf{X}_{r}\bm{%
\beta }-\mathbf{Z}_{r}\mathbf{u}_{r}\right) ^{\prime }|\mathbf{y}_{r}\right]
\mathbf{X}_{r}, \\
B_{\bm{\beta }\phi _{gh}} &=&-\sum_{r=1}^{R}E\left[ u_{gr}u_{hr}\mathbf{X%
}_{r}^{\prime }\left( \mathbf{y}_{r}^{\ast }-\mathbf{X}_{r}\bm{\beta }-%
\mathbf{Z}_{r}\mathbf{u}_{r}\right) |\mathbf{y}_{r}\right] \\
&&+\frac{1}{2}\sum_{r=1}^{R}E\left( u_{gr}u_{hr}|\mathbf{y}\right) \mathbf{X}%
_{r}^{\prime }E\left[ \left( \mathbf{y}_{r}^{\ast }-\mathbf{X}_{r}\bm{%
\beta }-\mathbf{Z}_{r}\mathbf{u}_{r}\right) |\mathbf{y}_{r}\right] , \\
B_{\phi _{gh}\phi _{k\ell }} &=&-\frac{R}{2}Tr\left( \mathbf{\Phi }_{G}^{-1}%
\mathbf{J}^{gh}\mathbf{\Phi }_{G}^{-1}\mathbf{J}^{k\ell }\right)
+\sum_{r=1}^{R}E\left( u_{gr}u_{hr}u_{kr}u_{h\ell }|\mathbf{y}_{r}\right) \\
&&-\sum_{r=1}^{R}E\left( u_{gr}u_{hr}|\mathbf{y}_{r}\right) E\left(
u_{kr}u_{h\ell }|\mathbf{y}_{r}\right) .
\end{eqnarray*}%
The above expressions imply computation of the third and fourth central
moments of $\mathbf{u}_{r}^{\prime }\mathbf{|y}_{r}$. To simplify
computations, we approximate $E\left( \mathbf{u}_{r}\mathbf{|y}_{r}\right) $
and $E\left( \mathbf{u}_{r}\mathbf{u}_{r}^{\prime }\mathbf{|y}_{r}\right) $
by (\ref{Eu})-(\ref{Euu}). Only for these derivations we also approximate
the third and fourth central moments of $u_{r}\mathbf{|y}_{r}$ by those of $%
\bar{y}_{ir}^{\ast }\mathbf{|y}_{r}$. Let $\mathbf{M_{r}}^{-1}\mathbf{Z}_{r}^{\prime
}E\left[ \left( \mathbf{y}_{r}^{\ast }-\mathbf{X}_{r}\bm{\beta }\right) |%
\mathbf{y}^{\ast }\right] $ be a vector with generic, $i$th element given by
$E\left[ \left( \bar{y}_{ir}^{\ast }-\bm{\beta }^{\prime }\mathbf{\bar{x}%
}_{ir}\right) |\mathbf{y}_{r}\right] $. Noting that%
\begin{equation*}
\mathbf{y}_{r}^{\ast }-\mathbf{X}_{r}\bm{\beta }-\mathbf{Z}_{r}\mathbf{M_{r}}^{-1}%
\mathbf{Z}_{r}^{\prime }\left( \mathbf{y}_{r}^{\ast }-\mathbf{X}_{r}\bm{%
\beta }\right) =\left( \mathbf{I}_{N_{r}}-\mathbf{Z}_{r}\mathbf{M_{r}}^{-1}\mathbf{Z}%
_{r}^{\prime }\right) \left( \mathbf{y}_{r}^{\ast }-\mathbf{X}_{r}\bm{%
\beta }\right) ,
\end{equation*}%
$B_{\bm{\beta \beta }}$ becomes%
\begin{eqnarray*}
B_{\bm{\beta \beta }} &=&-\sum_{r=1}^{R}\mathbf{X}_{r}^{\prime }\mathbf{X%
}_{r}+\sum_{r=1}^{R}\mathbf{X}_{r}^{\prime }\left( \mathbf{I}_{N_{r}}-%
\mathbf{Z}_{r}\mathbf{M_{r}}^{-1}\mathbf{Z}_{r}^{\prime }\right) E\left[ \left(
\mathbf{y}_{r}^{\ast }-\mathbf{X}_{r}\bm{\beta }\right) \left( \mathbf{y}%
_{r}^{\ast }-\mathbf{X}_{r}\bm{\beta }\right) ^{\prime }|\mathbf{y}_{r}%
\right] \left( \mathbf{I}_{N_{r}}-\mathbf{Z}_{r}\mathbf{M_{r}}^{-1}\mathbf{Z}%
_{r}^{\prime }\right) ^{\prime }\mathbf{X}_{r} \\
&&-\sum_{r=1}^{R}\mathbf{X}_{r}^{\prime }\left( \mathbf{I}_{N_{r}}-\mathbf{Z}%
_{r}\mathbf{M_{r}}^{-1}\mathbf{Z}_{r}^{\prime }\right) E\left[ \left( \mathbf{y}%
_{r}^{\ast }-\mathbf{X}_{r}\bm{\beta }\right) |\mathbf{y}_{r}\right] E%
\left[ \left( \mathbf{y}_{r}^{\ast }-\mathbf{X}_{r}\bm{\beta }\right)
^{\prime }|\mathbf{y}_{r}\right] \left( \mathbf{I}_{N_{r}}-\mathbf{Z}%
_{r}\mathbf{M_{r}}^{-1}\mathbf{Z}_{r}^{\prime }\right) ^{\prime }\mathbf{X}_{r} \\
&=&-\sum_{r=1}^{R}\mathbf{X}_{r}^{\prime }\mathbf{X}_{r}+\sum_{r=1}^{R}%
\mathbf{X}_{r}^{\prime }\left( \mathbf{I}_{N_{r}}-\mathbf{Z}_{r}\mathbf{M_{r}}^{-1}%
\mathbf{Z}_{r}^{\prime }\right) \mathbf{\mathbf{\Sigma }}_{r}\left( \mathbf{I%
}_{N_{r}}-\mathbf{Z}_{r}\mathbf{M_{r}}^{-1}\mathbf{Z}_{r}^{\prime }\right) \mathbf{X}%
_{r}.
\end{eqnarray*}%
Note that, under our approximation, $\mathbf{u}_{r}^{\prime }\mathbf{J}^{gh}%
\mathbf{u}_{r}=\left( \mathbf{y}_{r}^{\ast }-\mathbf{X}_{r}\bm{\beta }%
\right) ^{\prime }\mathbf{A}^{gh}\left( \mathbf{y}_{r}^{\ast }-\mathbf{X}_{r}%
\bm{\beta }\right) $, and $B_{\bm{\beta }\phi _{gh}}$ becomes:%
\begin{eqnarray*}
B_{\bm{\beta }\phi _{gh}} &=&-\frac{1}{2}\sum_{r=1}^{R}\mathbf{X}%
_{r}^{\prime }\left( \mathbf{I}_{N_{r}}-\mathbf{Z}_{r}\mathbf{M_{r}}^{-1}\mathbf{Z}%
_{r}^{\prime }\right) E\left[ \left( \mathbf{y}_{r}^{\ast }-\mathbf{X}_{r}%
\bm{\beta }\right) ^{\prime }\mathbf{A}^{gh}\left( \mathbf{y}_{r}^{\ast
}-\mathbf{X}_{r}\bm{\beta }\right) \left( \mathbf{y}_{r}^{\ast }-\mathbf{%
X}_{r}\bm{\beta }\right) |\mathbf{y}_{r}\right] \\
&&+\frac{1}{2}\sum_{r=1}^{R}\mathbf{X}_{r}^{\prime }\left( \mathbf{I}%
_{N_{r}}-\mathbf{Z}_{r}\mathbf{M_{r}}^{-1}\mathbf{Z}_{r}^{\prime }\right) E\left[
\left( \mathbf{y}_{r}^{\ast }-\mathbf{X}_{r}\bm{\beta }\right) ^{\prime }%
\mathbf{A}^{gh}\left( \mathbf{y}_{r}^{\ast }-\mathbf{X}_{r}\bm{\beta }%
\right) |\mathbf{y}_{r}\right] E\left[ \left( \mathbf{y}_{r}^{\ast }-\mathbf{%
X}_{r}\bm{\beta }\right) |\mathbf{y}_{r}\right] ,
\end{eqnarray*}%
Using (\ref{ex7}) Corollary \ref{cor1}, it follows that $B_{\bm{\beta }%
\phi _{gh}}$ is:%
\begin{equation*}
B_{\bm{\beta }\phi _{gh}}=-\frac{1}{2}\sum_{r=1}^{R}\mathbf{X}%
_{r}^{\prime }\left( \mathbf{I}_{N_{r}}-\mathbf{Z}_{r}\mathbf{M_{r}}^{-1}\mathbf{Z}%
_{r}^{\prime }\right) \left[ diag(\mathbf{\Lambda }_{3}\mathbf{\Sigma }%
_{r}^{1/2}\mathbf{A}^{gh}\mathbf{\mathbf{\Sigma }}_{r}^{1/2}\mathbf{)\mathbf{%
\Sigma }}_{r}^{1/2}\mathbf{1}_{n}+2\mathbf{\mathbf{\Sigma }}_{r}\mathbf{A}%
^{gh}E\left( \left( \mathbf{y}_{r}^{\ast }-\mathbf{X}_{r}\bm{\beta }%
\right) |\mathbf{y}_{r}\right) \right] .
\end{equation*}%
We now focus on $B_{\phi _{gh}\phi _{k\ell }}$. Using (\ref{ex5})
\begin{eqnarray*}
&&E\left[ \left( \mathbf{y}_{r}^{\ast }-\mathbf{X}_{r}\bm{\beta }\right)
^{\prime }\mathbf{A}^{gh}\left( \mathbf{y}_{r}^{\ast }-\mathbf{X}_{r}\bm{%
\beta }\right) \left( \mathbf{y}_{r}^{\ast }-\mathbf{X}_{r}\bm{\beta }%
\right) ^{\prime }\mathbf{A}^{k\ell }\left( \mathbf{y}_{r}^{\ast }-\mathbf{X}%
_{r}\bm{\beta }\right) |\mathbf{y}_{r}\right] \\
&&-E\left[ \left( \mathbf{y}_{r}^{\ast }-\mathbf{X}_{r}\bm{\beta }%
\right) ^{\prime }\mathbf{A}^{gh}\left( \mathbf{y}_{r}^{\ast }-\mathbf{X}_{r}%
\bm{\beta }\right) |\mathbf{y}_{r}\right] E\left[ \left( \mathbf{y}%
_{r}^{\ast }-\mathbf{X}_{r}\bm{\beta }\right) ^{\prime }\mathbf{A}%
^{k\ell }\left( \mathbf{y}_{r}^{\ast }-\mathbf{X}_{r}\bm{\beta }\right) |%
\mathbf{y}_{r}\right] \\
&=&Tr\left[ \mathbf{\Lambda }_{4}\mathbf{\Sigma }_{r}^{1/2}\mathbf{A}^{k\ell
}\mathbf{\Sigma }_{r}^{1/2}diag\left( \mathbf{\mathbf{\Sigma }}_{r}^{1/2}%
\mathbf{A}^{gh}\mathbf{\Sigma }_{r}^{1/2}\right) \right] +2Tr\left( \mathbf{%
\Sigma }_{r}\mathbf{\mathbf{A}}^{k\ell }\mathbf{\Sigma }_{r}\mathbf{A}%
^{gh}\right) \\
&&+2\left( diag(\mathbf{\Lambda }_{3}\mathbf{\Sigma }_{r}^{1/2}\mathbf{A}%
^{gh}\mathbf{\Sigma }_{r}^{1/2}\mathbf{)1}_{n}\right) ^{\prime }\mathbf{%
\Sigma }_{r}^{1/2}\mathbf{A}^{k\ell }E\left[ \left( \mathbf{y}_{r}^{\ast }-%
\mathbf{X}_{r}\bm{\beta }\right) |\mathbf{y}_{r}\right] \\
&&+2\left( diag(\mathbf{\Lambda }_{3}\mathbf{\Sigma }_{r}^{1/2}\mathbf{A}%
^{k\ell }\mathbf{\mathbf{\Sigma }}_{r}^{1/2}\mathbf{)1}_{n}\right) ^{\prime }%
\mathbf{\Sigma }_{r}^{1/2}\mathbf{A}^{gh}E\left[ \left( \mathbf{y}_{r}^{\ast
}-\mathbf{X}_{r}\bm{\beta }\right) |\mathbf{y}_{r}\right] \\
&&4E\left[ \left( \mathbf{y}_{r}^{\ast }-\mathbf{X}_{r}\bm{\beta }%
\right) |\mathbf{y}_{r}\right] ^{\prime }\mathbf{A}^{gh}\mathbf{\Sigma }_{r}%
\mathbf{A}^{k\ell }E\left[ \left( \mathbf{y}_{r}^{\ast }-\mathbf{X}_{r}%
\bm{\beta }\right) |\mathbf{y}_{r}\right] ,
\end{eqnarray*}%
so that we obtain%
\begin{eqnarray*}
B_{\phi _{gh}\phi _{k\ell }} &=&-\frac{R}{2}Tr\left( \mathbf{\Phi }_{G}^{-1}%
\mathbf{J}^{gh}\mathbf{\Phi }_{G}^{-1}\mathbf{J}^{k\ell }\right) +\frac{1}{4}%
\sum_{r=1}^{R}2Tr\left( \mathbf{\Sigma }_{r}\mathbf{\mathbf{A}}^{k\ell }%
\mathbf{\Sigma }_{r}\mathbf{A}^{gh}\right) \\
&&+\frac{1}{4}\sum_{r=1}^{R}Tr\left[ \mathbf{\Lambda }_{4}\mathbf{\Sigma }%
_{r}^{1/2}\mathbf{A}^{k\ell }\mathbf{\Sigma }_{r}^{1/2}diag\left( \mathbf{%
\mathbf{\Sigma }}_{r}^{1/2}\mathbf{A}^{gh}\mathbf{\Sigma }_{r}^{1/2}\right) %
\right] \\
&&+\frac{1}{4}\sum_{r=1}^{R}2\left( diag(\mathbf{\Lambda }_{3}\mathbf{\Sigma
}_{r}^{1/2}\mathbf{A}^{gh}\mathbf{\Sigma }_{r}^{1/2}\mathbf{)1}_{n}\right)
^{\prime }\mathbf{\mathbf{\Sigma }}_{r}^{1/2}\mathbf{A}^{k\ell }E\left[
\left( \mathbf{y}_{r}^{\ast }-\mathbf{X}_{r}\bm{\beta }\right) |\mathbf{y%
}_{r}\right] \\
&&+\frac{1}{4}\sum_{r=1}^{R}2\left( diag(\mathbf{\Lambda }_{3}\mathbf{\Sigma
}_{r}^{1/2}\mathbf{A}^{k\ell }\mathbf{\mathbf{\Sigma }}_{r}^{1/2}\mathbf{)1}%
_{n}\right) ^{\prime }\mathbf{\mathbf{\Sigma }}_{r}^{1/2}\mathbf{A}^{gh}E%
\left[ \left( \mathbf{y}_{r}^{\ast }-\mathbf{X}_{r}\bm{\beta }\right) |%
\mathbf{y}_{r}\right] \\
&&+\frac{1}{4}\sum_{r=1}^{R}4E\left[ \left( \mathbf{y}_{r}^{\ast }-\mathbf{X}%
_{r}\bm{\beta }\right) |\mathbf{y}_{r}\right] ^{\prime }\mathbf{A}^{gh}%
\mathbf{\mathbf{\Sigma }}_{r}^{1/2}\mathbf{A}^{k\ell }E\left[ \left( \mathbf{%
y}_{r}^{\ast }-\mathbf{X}_{r}\bm{\beta }\right) |\mathbf{y}_{r}\right] .
\end{eqnarray*}
\renewcommand{\thesection}{2}
\end{document}